\documentclass[dvips,ssy,preprint]{imsart}

\RequirePackage[OT1]{fontenc}
\RequirePackage{amsthm,amsmath}
\RequirePackage[numbers]{natbib}
\RequirePackage[colorlinks,citecolor=blue,urlcolor=blue]{hyperref}

\usepackage[T1]{fontenc}
\usepackage{amsmath}
\usepackage{amssymb}
\usepackage{mathrsfs}
\usepackage{amsthm}
\usepackage{graphicx}

\arxiv{arXiv:1202.2089}

\startlocaldefs
\numberwithin{equation}{section}
\theoremstyle{plain}
\newtheorem{lemma}{Lemma}\newtheorem{theorem}{Theorem}\newtheorem{corollary}{Corollary}\newtheorem{remark}{Remark}
\newtheorem{definition}{Definition}
\newcommand{\post}[2]{\begin{center} \includegraphics[width=#2]{#1} \end{center} }
\endlocaldefs

\begin{document}

\begin{frontmatter}
\title{The Supermarket Game \thanksref{T1,T2} }
\runtitle{The Supermarket Game}
\thankstext{T1}{This research was supported by the Network Science Foundation under Grant CCF 10-16959.}
\thankstext{T2}{Revised December 2012. A preliminary version of this paper appeared in the Proceedings of the 2012 International Symposium on Information Theory.}


\begin{aug}
\author{\fnms{Jiaming} \snm{Xu} \ead[label=e1]{jxu18@illinois.edu} }
\and
\author{\fnms{Bruce} \snm{Hajek} \ead[label=e2]{b-hajek@uiuc.edu} }
\runauthor{J. XU AND B. HAJEK}
\affiliation{University of Illinois at Urbana-Champaign}

\address{Coordinated Science Laboratory,\\
 University of Illinois at Urbana-Champaign, \\
 Urbana, IL, 61801, USA \\
 \printead{e1} \\
\phantom{E-mail:\ }\printead*{e2} }
\end{aug}

\begin{abstract}
A supermarket game is considered with $N$ FCFS queues with unit exponential service rate and global Poisson arrival rate $N \lambda$.
Upon arrival each customer chooses a number of queues to be sampled uniformly at random and joins the least loaded sampled queue.
Customers are assumed to have cost for both waiting and sampling, and they want to minimize their own expected total cost.

We study the supermarket game in a mean field model that corresponds to the limit as $N$ converges to infinity in the sense that
(i) for a fixed symmetric customer strategy, the joint equilibrium distribution of any fixed number of queues converges as $N \to \infty$
to a product distribution determined by the mean field model and (ii) a Nash equilibrium for the mean field model is an $\epsilon$-Nash equilibrium
for the finite $N$ model with $N$ sufficiently large. It is shown that there always exists a Nash equilibrium for $\lambda <1$ and the Nash equilibrium
is unique with homogeneous waiting cost for $\lambda^2 \le 1/2$. Furthermore, we find that the action of sampling more queues by some customers
has a positive externality on the other customers in the mean field model, but can have a negative externality for finite $N$.
\end{abstract}

\begin{keyword}[class=AMS]
\kwd[Primary ]{60K35}
\kwd[; secondary ]{60K25,91A40}
\end{keyword}

\begin{keyword}
\kwd{Mean Field Model}
\kwd{Queueing}
\kwd{Nash Equilibrium}
\kwd{Externality}
\end{keyword}

\end{frontmatter}

\section{Introduction}
Consider a stream of customers arriving to a multi-server system where any server is capable of serving any customer. Upon arrival, customers are unaware of the current queue length at servers, so they sample a few servers and join the server with the shortest queue among the sampled few. Customers have time cost proportional to the waiting time at servers and sampling cost proportional to the number of sampled servers. Customers are self-interested and aim to minimize their own total cost by choosing the optimal number of servers to sample. Note that the waiting time of a customer depends on the other customers' choices, so it is a game among the customers and we call it the {\it supermarket game}, because often in supermarkets customers try to find counters with short queue to check out.

\subsection{Motivation}
The supermarket game is a simple model for analyzing distributed load balancing in transportation and communication networks. Load balancing ensures efficient resource utilization and improves the quality of service, by evenly distributing the workload across multiple servers. Traditionally, load balancing is fulfilled by a central dispatcher that assigns the newly arriving work to the server with the least workload.
As modern or future networks become larger and increasingly distributed, such central dispatcher may not exist, and thus the load balancing has to be carried out by customers themselves. Hence, the supermarket game is relevant in scenarios where (1) customers choose which server to join without directions from a central dispatcher or tracker; (2) global workload or queue length information is not available and customers randomly choose a finite number of servers to probe; (3) there is cost associated with probing a server and waiting in a queue.

Examples of such scenarios are the following:
\begin{itemize}
\item Network routing: customers represent traffic flows and servers represent possible routes from a given source to a destination. A traffic flow can find the route with low delay by probing different routes.

\item Dynamic wireless spectrum access: customers represent wireless devices and servers represent all the shared spectrum. The wireless devices can find the spectrum band with low interference and congestion by probing multiple spectrum bands.

\item Cloud computing service: customers can decide how many servers to probe in seeking the server with low delay.
\end{itemize}
In this paper, we address the following natural questions for these systems: How many servers will a self-interested customer sample? Is sampling or probing more servers by some customers beneficial or detrimental to the others?

\subsection{Main Results}
The supermarket game with finite number of servers is difficult to analyze due to the correlation among queues at different servers. Therefore, we study the supermarket game in a mean field model that corresponds to the limit as the number of servers converges to infinity. By assuming: (1) unit exponential service rate at servers; (2) Poisson arrival of customers; (3) homogeneous waiting cost and sampling cost, it is shown that:
\begin{itemize}
\item There exists a mixed strategy Nash equilibrium for all arrival rates per server less than one.
\item The action of sampling more servers by some customers has a positive externality on the other customers, which further implies that customers sample no more queues for any Nash equilibrium than for the socially optimal strategy.
\item Nash equilibrium is unique for arrival rates per server less than or equal to $1/\sqrt{2}$.
\item Nash equilibrium is unique if and only if a local monotonicity condition is satisfied. This condition is used to explore the uniqueness numerically for arrival rates per server larger than $1/\sqrt{2}$.
\item Multiple Nash equilibria exist for a particular example with arrival rates per server equal to $0.999$.
\item Nash equilibrium is unique for arrival rates per server equal to $0.999$ if customers can only sample either one queue or two queues.
\end{itemize}
Then, we consider the heterogenous waiting cost case and prove the existence of pure strategy Nash equilibrium in that case.

We also show that the mean field model arises naturally as the limit of supermarket game with finite number of servers:
\begin{itemize}
\item A propagation of chaos result and a coupling result similar to the ones in \citep{Graham00} and \citep{Turner98} hold, and thus the joint equilibrium distribution of any fixed number of queues converges to a product distribution determined by the mean field model as the number of queues converge to infinity.
\item A Nash equilibrium of the supermarket game in the mean field model is an $\epsilon$-Nash equilibrium of the supermarket game for finite number of servers with the number of servers large enough.
\end{itemize}
Furthermore, in the supermarket game with finite number of servers, we find that sampling more queues by some customers has a nonnegative externality on customers who only sample one queue, but it can have a {\it negative} externality for customers sampling more than one queue, which is in sharp contrast to the conclusion in the mean field model.

\subsection{Related Work}
The supermarket game is formulated based on the classical supermarket model with $N$ parallel queues in which customers sample a fixed number $L$ of queues uniformly at random and join the shortest sampled queue. The supermarket model has been extensively studied in the literature using the mean field approach.
Vvedenskaya et. al. \citep{VveDobKar96} shows that in the mean field model, the equilibrium queue sizes decay doubly exponentially for $L\ge 2$. Turner \citep{Turner98} proves an interesting coupling result for fixed $N$, which implies that the load in the network is handled better and more evenly as $L$ increases. Graham \citep{Graham00} proves a chaoticity result on the path space using the idea of propagation of chaos \citep{Sznitman91}, and Mitzenmacher \citep{Mitzenmacher00} studies the model using Kurtz's theorem. Furthermore, Luczak and McDiarmid \citep{McDiarmid00} prove that for $L \ge 2$, the maximum queue length scales as $\ln \ln n /\ln L + O(1)$. The recent paper by Bramson et. al. \citep{Bramson10} analyzes the supermarket model with general service time distributions. They show that in the case of FCFS discipline and power-law service time, the equilibrium queue sizes will decay doubly exponentially, exponentially, or just polynomially, depending on the power-law exponent and the number of choices $L$. Ganesh et. al. \citep{Ganesh10} studies a variant of the supermarket model, where the customers initially join an arbitrary server, but may switch to other servers later independently at random. They find that in the mean field model, the average waiting time under the load-oblivious switching strategy is not considerably larger than that under a smarter load-aware switching strategy.

In addition to the supermarket model, the mean field approach has also been used to analyze scheduling in queueing networks, such as the CSMA algorithm in a wireless local area network \citep{Proutiere10} and downlink transmission scheduling \citep{Alanyali11}. Also, a recent work \citep{Tsitsiklis11,Tsitsiklis12} investigates the performance tradeoff between centralized and distributed scheduling in a multi-server system for the mean field model. However, none of above work considers a game-theoretic framework.

The supermarket game proposed in this paper falls into a large research area involving equilibrium behavior of customers and servers, known as {\it queueing games}. A comprehensive survey can be found in \citep{Haviv09}. A particularly relevant paper by Hassin and Haviv \citep{Haviv94} studies a two line queueing system, where upon arrival each customer decides whether to purchase the information about which line is shorter, or randomly select one of the lines. It shows how to find a Nash equilibrium and examines the externality imposed by an informed customer on the others. A model in which customers can balk, either before or after sampling the backlog at a single server queue, is considered in \citep{Hassin12}. the paper finds benefits to the service operator of offering the possibility to balk after the backlog is sampled.

Finally, the supermarket game is also related to and partially motivated by the theory of {\it mean field games} in the context of dynamical games \citep{Huang07,Lions07}. The mean field game approach studies a weakly coupled, $N$ player game by letting $N \to \infty$. However, we caution that in the supermarket game, we consider an infinite sequence of customers instead of finitely many customers, which is different from an $N$ player game.

\subsection{Organization of the Paper}
Section \ref{SectionModel} introduces the precise definition of the supermarket game to be studied and the key notations. The  mean field model for the supermarket game is studied in Section \ref{SectionMFL}, and justified in Section \ref{SectionFiniteN}. Section \ref{SectionExternalityFinite} examines the externality of sampling more queues in the finite $N$ model. Section \ref{SectionConclusion} ends the paper with concluding remarks. Miscellaneous details and proofs are in the Appendix.

\section{Model and Notation} \label{SectionModel}
\subsection{Model}

Consider a supermarket game with $N$ FCFS queues with exponential service rate one, and global Poisson arrival rate $N \lambda$. Assume $\lambda <1$ and let $\mathcal{L}=\{1, \ldots, L_{\rm max} \}$. Upon arrival, each customer can choose a number $L \in \mathcal{L}$ of queues to be sampled uniformly at random, and the customer  joins the sampled queue with the least number of customers, ties being resolved uniformly at random. Customers are assumed to have cost $c$ per unit waiting time and $c_s$ for sampling one queue. These cost parameters are the same for all the customers; heterogeneous waiting costs are only considered in Section \ref{SectionHetero}.

For a fixed customer $i$, if she chooses $L_i$ queues to sample, and all the other customers choose $L_{-i}$, then by PASTA (Poisson arrival sees time average), the expected total cost of customer $i$ is given by
\begin{align}
C(L_i, L_{-i})=c \mathbb{E} [ W(L_i, L_{-i}) ] + c_s L_i,
\end{align}
where $\mathbb{E} [  W(L_i, L_{-i})  ]$ is the expected waiting time (service time included) under the stationary distribution. The goal of customer $i$ is to minimize her own expected total cost by choosing the optimal $L_i$.

Since the supermarket game is symmetric in the customers, we limit ourselves to symmetric strategies. We call $L^\star \in \mathcal{L}$ a pure strategy Nash equilibrium, if
\begin{align}
C(L^\star, L^\star) \le C(L_i, L^\star), \text{ for all } L_i \in \mathcal{L}. \nonumber
\end{align}
Since a pure strategy Nash equilibrium does not always exist, we are also interested in mixed strategy Nash equilibria.

Let $\mathcal{P}(\mathcal{L})$ denote the set of all probability distributions over $\mathcal{L}$.  The mixed strategy $\mu_i$ for customer $i$ is simply a probability distribution in $\mathcal{P}(\mathcal{L})$, i.e., $\mu_i(L_i)$ is the probability that customer $i$ samples $L_i$ queues. If all the other customers use the mixed strategy $\mu_{-i}$, then the expected total cost of customer $i$ using $\mu_i$ is given by
\begin{align}
C(\mu_i, \mu_{-i} )=\sum_{L_i=1}^{L_{\rm max}} C(L_i,\mu_{-i} ) \mu_i(L_i),
\end{align}
where $C(L_i,\mu_{-i})$ is the expected total cost of customer $i$ choosing $L_i$ given all the others choose the mixed strategy $\mu_{-i}$.

Define the best response correspondence for customer $i$ as $\text{BR}(\mu_{-i}):=\arg \min_{\mu_i} C(\mu_i, \mu_{-i})$. The correspondence $\text{BR}$ is a set-valued function from $\mathcal{P}( \mathcal{L})$ to subsets of $\mathcal{P}( \mathcal{L})$.
We call $\mu^\star \in \mathcal{P}( \mathcal{L})$ a mixed strategy Nash equilibrium if
\begin{align}
C(\mu^\star,\mu^\star) \le C(\mu_i, \mu^\star), \text{ for all } \mu_i \in \mathcal{P}( \mathcal{L}). \nonumber
\end{align}
In this paper, we are interested in characterizing the Nash equilibria of the supermarket game.

\subsection{Notation}
Let $\mathscr{X}$ denote a separable and complete metric space, $\mathcal{M}(\mathscr{X})$ and $\mathcal{P} (\mathscr{X})$ be the space of measures and space of probability measures on $\mathscr{X}$, respectively. In this paper, $\mathscr{X}$ will be $\mathcal{L}$, $\mathbb{N}$, $\mathbb{D}( \mathbb{R}_+, \mathbb{N} )$, or the space of probability measures on these spaces, where $\mathbb{N}$ is the set of natural numbers and $\mathbb{D}$ denotes the Skorokhod space. Let $\mathscr{L} (X)$ denote the law of a random variable $X$ on $\mathscr{X}$ and $\delta_{x}$ denote a point probability measure at $x \in \mathscr{X} $. Weak convergence of probability measures is denoted by $\Longrightarrow$.

Define on $\mathcal{M}(\mathscr{X})$ natural duality brackets with $L^\infty (\mathscr{X})$ as: for $\phi \in L^\infty (\mathscr{X})$ and $\mu \in \mathcal{M}(\mathscr{X})$, $\langle \phi, \mu \rangle = \int \phi d \mu$. Without ambiguity, brackets are also used to denote the quadratic covariation of continuous time martingales: let $M_1, M_2$ be two continuous time martingales, then $\langle M_1, M_2 \rangle$ is the standard quadratic covariation process. Define the total variation norm on $ \mathcal{M}(\mathscr{X})$ as $\| \mu \|_{\rm TV} = \sup \{ \langle f, \mu \rangle: \|f\|_{\infty} \le 1 \}$. Let $\| \mu_1 - \mu_2 \|_{\rm TV}$ be the corresponding total variation distance.

For $\mu_1, \mu_2 \in \mathcal{P}(\mathcal{L})$, use $ \mu_1 \le_{\rm st} \mu_2$ to denote that $\mu_1$ is first-order stochastically dominated by $\mu_2$, i.e., $\sum_{j=l}^{L_{\max}} \mu_1(j) \le \sum_{j=l}^{L_{\max}} \mu_2(j), \forall l \in \mathcal{L}$. For $x \in \mathbb{R}$, let $\lfloor x \rfloor$ denote the maximum integer no larger than $x$.

\section{Supermarket Game in a Mean Field Model} \label{SectionMFL}
The supermarket game in a mean field model is studied in this section by investigating the mixed strategy Nash equilibrium and the externality of sampling more queues.
\subsection{Mean Field Model}
In this subsection, we derive an expression for the expected total cost incurred by a customer in a mean field model, by assuming that queue lengths (including customers in service) are independent and identically distributed.

Suppose all the customers except customer $i$ use the mixed strategy $\mu_{-i}$, and let $r_t(k)$ denote the fraction of queues with at least $k$ customers at time $t$ in the mean field model. Then, the mean field equation is given by
\begin{align}
\frac{d r_t(k) }{d t} = \sum_{l=1}^{L_{\rm max}}  \lambda \mu_{-i}(l) \left( r_t^{l}(k-1) - r_t^{ l } (k) \right) - (r_t(k)-r_t(k+1) ),\label{EqMeanFieldEquation}
\end{align}
which is rigorously derived in Section \ref{SectionFiniteN}. For now, let us provide some intuition for each of the drift terms in (\ref{EqMeanFieldEquation}):

The term $\lambda \mu_{-i}(l) ( r_t^{l}(k-1) - r_t^{ l } (k) )$ corresponds to the arrivals of customers sampling $l$ queues. Because queue lengths are i.i.d,  $r_t^{l}(k-1)$ is the probability that the minimum queue length of $l$ uniformly sampled queues is greater than or equal to $k-1$. Thus, $( r_t^{l}(k-1) - r_t^{ l } (k) )$ is the probability that the minimum queue length of $l$ uniformly sampled queues is $k-1$, which is the same as the probability that a customer who samples $l$ queues joins a queue with $k-1$ customers. Note that $r_t(k)$ is increased if a customer joins a queue with $k-1$ customers. Therefore, $\sum_{l=1}^{L_{\rm max}}  \lambda \mu_{-i}(l) \left( r_t^{l}(k-1) - r_t^{ l } (k) \right)$ is the aggregate drift for $r_t(k)$ corresponding to arrivals.

The term $(r_t(k)-r_t(k+1) )$ corresponds to departures of customers at queues with exactly $k$ customers.

Since we are interested in the stationary regime, set $\frac{d r_t(k) }{d t} =0$ to yield equations for the equilibrium distribution denoted by $r_{\mu_{-i}}(k), k \ge 0$. For $k\ge 1$,
\begin{align}
\lambda \mathbb{E}_{\mu_{-i}} \left[ r_{\mu_{-i}}^{\mathbf{L}}(k-1) - r_{\mu_{-i}}^{\mathbf{L}} (k)   \right] = r_{\mu_{-i}}(k)-r_{\mu_{-i}}(k+1), \nonumber
\end{align}
where the random variable $\mathbf{L}$ is distributed as ${\mu_{-i}}$.
By summing the above equation for $k_0 \le k < + \infty$ using telescoping sums and changing $k_0$ to $k$, it follows that
\begin{align}
r_{\mu_{-i}}(0)=1,  r_{\mu_{-i}}(k) = \lambda u_{\mu_{-i}}( r_{{\mu_{-i}}}(k-1)), \label{EqDistriMeanFieldLimit}
\end{align}
where $u_{\mu_{-i}} (x):=\mathbb{E}_{\mu_{-i}} \left[ x^{\mathbf{L}} \right]= \sum_{l=1}^{L_{\max}} x^l {\mu_{-i}}(l)$.

Because queues are independent in the mean field model,
\begin{align}
\mathbb{E} [ W(L_i, \mu_{-i} ) ] = 1+  \mathbb{E} [ N(L_i) ] = 1+ \sum_{k=1}^\infty \mathbb{P} [  N (L_i) \ge k ] = \sum_{k=0}^\infty  r_{\mu_{-i}}^{L_i} (k), \nonumber
\end{align}
where $N(L_i)$ is the length of the shortest queue among the $L_i$ sampled queues. Therefore,
\begin{align}
C(L_i,\mu_{-i}) = c \sum_{k=0}^\infty  r_{\mu_{-i}}^{L_i}(k) +c_s L_i .
\end{align}
Since $g(L)=x^L$ is convex in $L \in \mathbb{R}$ for $x \ge 0$ and strictly convex for $0<x<1$, it follows that $\mathbb{E} [ W(L, \mu_{-i} ) ] $ and $C(L,\mu_{-i}) $ are strictly convex in $L \in \mathbb{R}$.

Next, we prove several key lemmas which are useful in the sequel.
\begin{lemma} \label{LemmaBestResponse}
The best response set $\text{BR}(\mu_{-i})$ consists of probability measures concentrated on an integer or two consecutive integers.
\end{lemma}
\begin{proof}
Let $\mu_i \in \text{BR}(\mu_{-i})$. Suppose there exists $L_1 < L_2 < L_3 \in \mathcal{L}$ such that $\mu_i(L_1)>0$ and $ \mu_i(L_3)>0$. Then, by the definition of $\text{BR}$,
\begin{align}
C(L_1, \mu_{-i})= C(L_3, \mu_{-i}) \le C(L_2, \mu_{-i}), \nonumber
\end{align}
which contradicts the fact that $C(L,\mu_{-i})$ is strictly convex in $L$ and thus the conclusion follows.
\end{proof}
\begin{remark}
Lemma \ref{LemmaBestResponse} implies that a probability measure $\mu \in \text{BR}(\mu_{-i})$ can be identified with a unique real number $L \in [1,L_{\max}]$. A number $L \in [1,L_{\max}]$ with $L=\lfloor L \rfloor + p$ for some $ 0 \le p < 1$ is identified with the probability measure with mass $1-p$ at $\lfloor L \rfloor$ and $p$ at $\lfloor L \rfloor+1$. Thus, we use real numbers to refer to probability measures in $\text{BR}(\mu_{-i})$.
\end{remark}

The next lemma translates the stochastic dominance relations between strategies into the stochastic dominance relations between mean field equilibrium distributions.
\begin{lemma} \label{LemmaMonotonicity}
Fix any $\mu_1, \mu_2 \in \mathcal{P}(\mathcal{L})$ such that $\mu_1 \le_{\rm st} \mu_2$. Then, for all $k \in \mathbb{N}$,  $r_{\mu_1}(k) \ge r_{\mu_2}(k)$. Furthermore, if $\mu_1 <_{\rm st} \mu_2$, then for all $k \ge 2$, $r_{\mu_1}(k) > r_{\mu_2}(k)$. Also, it follows that for all $k \in \mathbb{N} $ and all $\mu \in\mathcal{P}(\mathcal{L}) $, $r_\mu(k) \le \lambda^k$ and $C(L,\mu)$ is bounded independently of $L$ and $\mu$.
\end{lemma}
\begin{proof}
Since $\mu_1 \le_{\rm st} \mu_2$, it follows that for all $x \in [0,1] $, $u_{\mu_1}(x) \ge u_{\mu_2}(x)$. We prove the lemma by induction. If $k=0$, then $r_{\mu_1}(0)=r_{\mu_2}(0)=1$ and $r_{\mu_1}(0) \ge r_{\mu_2}(0)$ trivially holds.
If $r_{\mu_1}(k-1) \ge r_{\mu_2}(k-1)$, then
\begin{align}
r_{\mu_1}(k) &= \lambda u_{\mu_1} \left[ r_{\mu_1} (k-1)  \right] \ge  \lambda u_{\mu_1} \left[ r_{\mu_2}(k-1)  \right] \nonumber \\
& \ge \lambda u_{\mu_2}  \left[ r_{\mu_2} (k-1)  \right] = r_{\mu_2}(k). \nonumber
\end{align}
Therefore, for all $k \in \mathbb{N}$, $r_{\mu_1}(k) \ge r_{\mu_2}(k)$. Moreover, if $ \mu_1 <_{\rm st} \mu_2$, it follows that for all $ x \in (0,1) $, $u_{\mu_1}(x) > u_{\mu_2}(x)$. Thus,  for all $ k \ge 2$,  $r_{\mu_1}(k) > r_{\mu_2}(k)$.

Let $\mu_1=\delta_{1}$ be a point measure at singleton $1$. Then, $r_{\mu_1}(k)=\lambda^k$. Because for any $\mu \in \mathcal{P}(\mathcal{L})$, $\delta_{1} \le_{\rm st} \mu$, then $r_\mu(k) \le \lambda^k$ and $C(L,\mu) \le c/(1-\lambda) + c_s L_{\max}$.
\end{proof}

The next lemma states that $C(\mu_i,\mu_{-i})$ is a continuous function.
\begin{lemma} \label{LemmaContinuity}
$C(\mu_i,\mu_{-i})$ is jointly continuous with respect to $\mu_i$ and $\mu_{-i}$.
\end{lemma}
\begin{proof}
See the proof in Appendix \ref{ProofLemmaContinuity}.
\end{proof}

\subsection{Existence and Uniqueness of Nash Equilibrium}
In this subsection, we show the existence of a mixed strategy Nash equilibrium. It is easy to see that if there exists $\mu^\star$ with $\mu^\star \in \text{BR}(\mu^\star)$, i.e., $\mu^\star$ is a fixed point of the best response correspondence, then $\mu^\star$ is a mixed strategy Nash equilibrium. Thus, it suffices to show the existence of such a fixed point. The Kakutani fixed point theorem is used to prove it. \\

\noindent{\bf(Kakutani's Theorem)} Let $\mathcal{S}$ be a nonempty, compact and convex subset of some Euclidean space $\mathbb{R}^n$. Let $g: \mathcal{S} \to 2^{\mathcal{S}}$ be a set-valued function on $\mathcal{S}$ with a closed graph and the property that $g(x)$ is a nonempty and convex set for all $x \in \mathcal{S}$. Then $g$ has a fixed point.

In our setting, $\mathcal{S}=\mathcal{P}( \mathcal{L})$ and $g=\text{BR}$. It is known that $\mathcal{P}( \mathcal{L})$, as a space of probability distribution on finite set, is a nonempty, compact and convex subset of some Euclidean space. Also, since $C(\mu_i,\mu_{-i})$ is continuous, by Weierstrass Theorem, $\text{BR}(\mu_{-i})$ is a nonempty set. In addition, $C(\mu_i,\mu_{-i})$ is linear in $\mu_i$, so $\text{BR}(\mu_{-i})$ is a convex set. The last and key step is to prove that $\text{BR}$ has a closed graph, i.e., suppose $\mu_{-i}^{(n)} \to \mu_{-i}$, $\mu_i^{(n)} \in \text{BR}(\mu_{-i}^{(n)})$ and $\mu_i^{(n)} \to \mu_i$, we need to show that $\mu_i \in \text{BR} (\mu_{-i})$. It is proved in the following theorem using the continuity of $C(\mu_i,\mu_{-i})$.

\begin{theorem} \label{TheoremExistenceNE}
The supermarket game has a mixed strategy Nash equilibrium in the mean field model.
\end{theorem}
\begin{proof}
See the proof in Appendix \ref{ProofTheoremExistenceNE}.
\end{proof}

The uniqueness of mixed strategy Nash equilibrium is proved next in case $\lambda^2 \le 1/2 $. Some definitions and lemmas are introduced first.

For customer $i$, her {\it marginal value of sampling} at an integer $L_i $ with all the others adopting $\mu_{-i}$, is defined as
\begin{align}
V(L_i, \mu_{-i} ):= \mathbb{E} [ W(L_i, \mu_{-i}) ]- \mathbb{E} [ W(L_i+1, \mu_{-i} ) ]. \nonumber
\end{align}
Intuitively, $V(L_i, \mu_{-i} )$ characterizes the reduction of expected waiting time when customer $i$ increases the number of sampled queues from $L_i$ to $L_i +1 $. For ease of notation, define $V(0,\mu_i)=\infty$ and $V(L_{\max}, \mu_i) =0$. Since $\mathbb{E} [W(L_i, \mu_{-i}) ]$ is strictly convex in $L \in \mathbb{R}$, $V(L_i, \mu_{-i} )$ is strictly decreasing in $L_i$. The following lemma characterizes the best response using $V(L_i, \mu_{-i} )$.

\begin{lemma} \label{LemmaBRMarginalValueInfo}
Fix any $\mu_{-i} \in \mathcal{P}(\mathcal{L})$ and any integer $L_i \in \mathcal{L}$. Then $L_i \in \text{BR} (\mu_{-i})$ if and only if
\begin{align}
V(L_i, \mu_{-i} )  \le c_s/c \le V(L_i-1, \mu_{-i} ). \label{EqMarginalValueInfo}
\end{align}
Furthermore, fix any non-integer $L_i \in [1,L_{\max}]$. Then $L_i \in \text{BR} (\mu_{-i})$ if and only if
\begin{align}
V (\lfloor L_i \rfloor , \mu_{-i} ) = c_s/c. \label{EqMarginalValueInfoSecond}
\end{align}
\end{lemma}
\begin{proof}
See the proof in Appendix \ref{ProofLemmaBRMarginalValueInfo}.
\end{proof}

The next lemma proves a global monotonicity property of $V(L_i, \mu_{-i} )$ in case $\lambda^2 \le 1/2$, which is useful for showing the uniqueness of Nash equilibrium in that case.
\begin{lemma} \label{LemmaMonUnique}
Assume $\lambda^2 \le 1/2 $ and fix any $\mu_{-i},\tilde{\mu}_{-i} \in \mathcal{P}(\mathcal{L})$ such that $\mu_{-i} <_{\rm st} \tilde{\mu}_{-i}$. Then $ V(L_i, \mu_{-i}) > V(L_i, \tilde{\mu}_{-i} )$ for $1 \le L_i \le L_{\max}-1$. Furthermore, let $L_{i} \in \text{BR}( \mu_{-i}  )$ and $\tilde{L}_i \in \text{BR}(\tilde{\mu}_{-i})$, then $\tilde{L}_i \le L_{i}$.
\end{lemma}
\begin{proof}
See the proof in Appendix \ref{ProofLemmaMonUnique}.
\end{proof}
\begin{remark}
The key ingredient in the proof of Lemma \ref{LemmaMonUnique} is to show that $r_{\mu_{-i}}(2) \le L_i/(L_i+1)$ for any $L_i$ and $\mu_{-i}$. By Lemma \ref{LemmaMonotonicity}, $r_{\mu_{-i}}(2) \le \lambda^2$. Hence, if $\lambda^2 \le 1/2$, then $r_{\mu_{-i}}(2) \le 1/2$. Meanwhile, $L_i/(L_i+1) \ge 1/2 $ for any integer $L_i$. Therefore, if $\lambda^2 \le 1/2$, then $r_{\mu_{-i}}(2) \le L_i/(L_i+1)$ for any $L_i$ and $\mu_{-i}$. The same argument is used in the next subsection.
\end{remark}

Lemma \ref{LemmaMonUnique} implies that for $\lambda^2 \le 1/2 $, a customer tends to sample fewer queues when all the other customers sample more queues. This is an instance of {\it avoid the crowd} behavior \citep{Haviv09}, leading to uniqueness of the Nash equilibrium.
\begin{theorem}
If $\lambda^2 \le \frac{1}{2}$, the supermarket game has a unique Nash equilibrium in the mean field model.
\end{theorem}
\begin{proof}
A Nash equilibrium exists by Theorem \ref{TheoremExistenceNE}. Let $L^\star, \tilde{L}^\star \in [1, L_{\max}]$ denote two possible Nash equilibria; we show that $L^\star = \tilde{L}^\star$.

Without loss of generality, suppose $L^\star < \tilde{L}^\star$; by Lemma \ref{LemmaMonUnique}, $L^\star \ge \tilde{L}^\star $, which is a contradiction and concludes the proof.
\end{proof}

However, the marginal value of sampling $V(L_i,L_{-i})$ is not always increasing in $L_{-i} \in [1, L_{\max}] $ ($L_{-i}$ refers to a probability distribution here) for all $\lambda <1$. By numerical results, we find that when $\lambda=0.99$ and $L_i=5$, $V(L_i,L_{-i})$ is increasing with $L_{-i}$ for $11 \le L_{-i} \le 16$, as shown in Fig.~\ref{FigureValueofInformationCounterExample}. This implies that sometimes a customer tends to sample more queues when all the others sample more queues. This {\it follow the crowd} behavior can lead to multiple Nash equilibria. Such a scenario is described in the next subsection.

\begin{figure}
\centering
\post{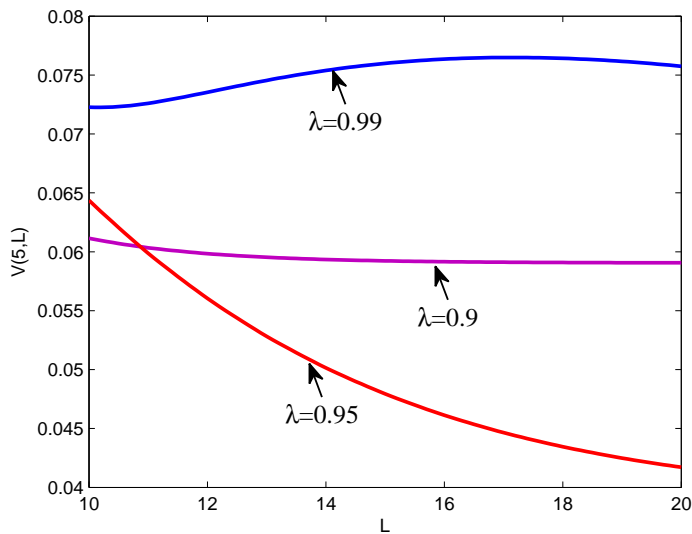}{4.0in}
\centering
\caption{The marginal value of sampling $V(5,L)$ with increasing $L$ from $10$ to $20$ and varying arrival rate $\lambda$.}
\label{FigureValueofInformationCounterExample}
\end{figure}

\subsection{Computation of Nash Equilibrium and Local Monotonicity Condition} \label{SectionComputationNE}
In this subsection, we show how to find a Nash equilibrium and establish the uniqueness of Nash equilibrium if a local monotonicity condition is satisfied. Then, a specific example where multiple Nash equilibria exist is constructed.

The following lemma determines a Nash equilibrium using the marginal value of sampling.
\begin{lemma}\label{LemmaCharacterizationNE}
Suppose $L^\star$ is determined by the following procedure:

\noindent (a) If $V(L_{\max}-1, L_{\max}) \ge  \frac{c_s}{c}$, set $L^\star = L_{\max}$. Otherwise

\noindent (b) Let $\hat{L}:= \min \{L \in \mathcal{L}: V(L, L+1) < \frac{c_s}{c} \}$,

\indent (b1) if $V(\hat{L},\hat{L}) \le \frac{c_s}{c}$, set $L^\star=\hat{L}$.

\indent (b2) if $V(\hat{L},\hat{L}) > \frac{c_s}{c} $, set $L^\star=\hat{L}+q^\star$, where $0<q^\star <1$ is the solution of $V(\hat{L},\hat{L}+q^\star)=\frac{c_s}{c}$.
Then, $L^\star$ is a Nash equilibrium for the supermarket game in the mean field model.
\end{lemma}
\begin{proof}
See the proof in Appendix \ref{ProofLemmaCharacterizationNE}.
\end{proof}

Next, we introduce a local monotonicity condition and show the uniqueness of Nash equilibrium for all values of $c$ and $c_s$ if and only if the local monotonicity condition is satisfied.
\begin{definition}
Given $0 <\lambda<1$ and any integer $L_{\max} \ge 1$, the local monotonicity condition is satisfied for $(\lambda, L_{\max})$ if $V(L,L+q)$ is strictly decreasing over $ 0\le q \le 1$ for each integer $L$ with $1 \le L \le L_{\max}-1$.
\end{definition}
\begin{theorem}\label{TheoremUniqNE}
Given $0 <\lambda<1$ and any integer $L_{\max} \ge 1$, the supermarket game in the mean field model has a unique Nash equilibrium for all values of $c$ and $c_s$ , if and only if the local monotonicity condition is satisfied for $(\lambda, L_{\max})$.
\end{theorem}
\begin{proof}
See the proof in Appendix \ref{ProofTheoremUniqNE}.
\end{proof}

The following numerical results, depicted in Fig.~\ref{FigureValueofInformationGeneralL}, show that when $\lambda=0.99$, $V(L,L+q)$ is indeed strictly decreasing with respect to $ 0\le q \le 1$ for $L=1, \ldots, 9$. For $L \ge 10$,
\begin{align}
 r_{L+q} (2) \le  r_L (2) \le r_{10} (2) =0.99^{(10^2-1)/(10-1)} \le 10/11 \le L/(L+1), \nonumber
\end{align}
which implies that $V(L,L+q)$ is strictly decreasing over $0 \le q \le 1$, in view of the proof for Lemma \ref{LemmaMonUnique}. Therefore, for $\lambda=0.99$ and any integer $L_{\max} \ge 1$, there exists a unique Nash equilibrium for all values of $c$ and $c_s$.
\begin{figure}
\centering
\post{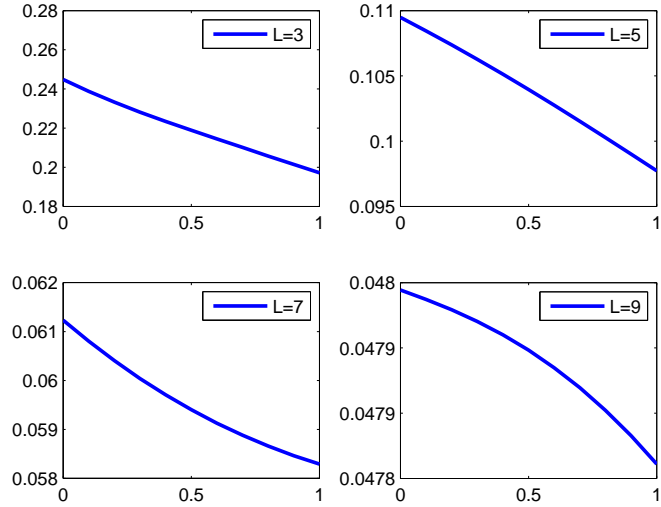}{4.0in}
\centering
\caption{For $\lambda=0.99$, the marginal value of sampling $V(L,L+q)$ with increasing $q$ from $0$ to $1$ and varying $L$.}
\label{FigureValueofInformationGeneralL}
\end{figure}
However, when $\lambda=0.999$ and $L_{\max}=25$,  $V(L,L+q)$ is strictly increasing in $q$ for $L=18,\ldots, 24$. Therefore, multiple Nash equilibria exist if $c_s/c=0.0148$, as depicted in Fig.~\ref{FigureMutipleNE}. We see that $L=19,\ldots, 24$ are pure strategy Nash equilibria and mixed strategy Nash equilibria exist between each two consecutive pure strategy Nash equilibria.
\begin{figure}
\centering
\post{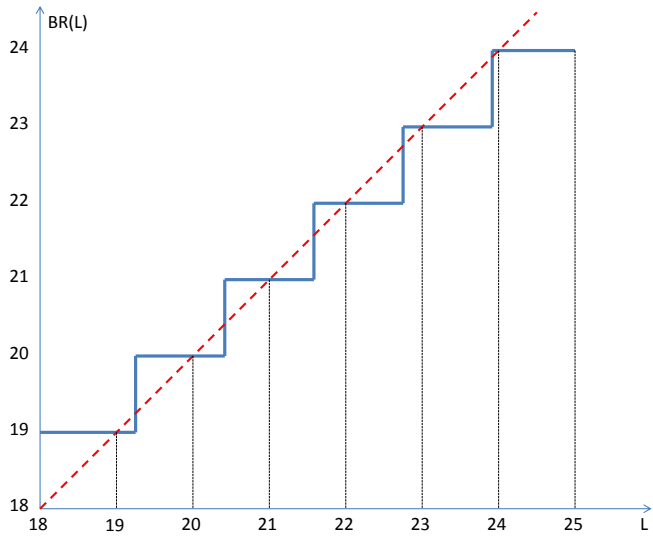}{4.0in}
\centering
\caption{The best response with respect to $L$ for $\lambda=0.999$ and $\frac{c_s}{c}=0.0148$.}
\label{FigureMutipleNE}
\end{figure}

\subsection{Special Case: Two Choices}
In this subsection, we consider a simple special case where all the customers only have two choices, $1$ or $2$.

Fixing a customer $i$, if $L_{-i} =1+q$, i.e., all the other customers choose one queue to sample with probability $1-q$ and two queues with probability $q$, then the stationary distribution in the mean field model can be derived as
\begin{align}
r_q(0)=1, r_q(k)= \lambda u_q(r_q(k-1)), k \ge 1, \nonumber
\end{align}
where $u_q(x)=x^2 q + x(1-q)$. Note that this stationary distribution has also been derived in Section 4.4.1 of\citep{Mitzenmacher96}.

Then, the total expected cost of customer $i$ choosing $1+p$ is given by
\begin{align}
 C(1+p,1+q) =(1-p) ( c \sum_{k=0}^\infty r_q(k) + c_s) + p(  c \sum_{k=0}^\infty r^2_q(k) + 2 c_s ). \nonumber
\end{align}
The marginal value of sampling  of customer $i$ at $1$ is given by
\begin{align}
V(1, 1+q)= \sum_{k=0}^\infty r_q(k)(1-r_q(k)). \nonumber
\end{align}
It follows that the best response under $L_{-i} =1+q$ is given by
\begin{align}
\text{BR}(1+q) = \left\{
\begin{array}{rl}
1 & \text{if } c_s/c > V(1, 1+q) ,\\
2 & \text{if } c_s/c < V(1, 1+q), \\
\left[1,2\right] & \text{if } c_s/c = V(1, 1+q). \label{EqBestResonpseTwoChoice}
\end{array} \right.
\end{align}

By numerical results depicted in Fig.~\ref{FigureValueofInformaiton}, we find that $V(1, 1+q)$ is strictly decreasing in $q$ for a sequence of $\lambda$ up to $0.999$. This is strong numerical evidence that the local monotonicity condition is satisfied for all $\lambda <1$. Therefore, by Theorem \ref{TheoremUniqNE} and Lemma \ref{LemmaBRMarginalValueInfo}, we {\it conjecture} that: (i) if $V(1,1) \le c_s/c$, then $L^\star=1$ is the unique Nash equilibrium; (ii) if $ V(1,2) \ge c_s/c$, then $L^\star=2$ is the unique Nash equilibrium; (iii) otherwise, there exists a $p^\star \in (0,1) $ such that $ V(1, 1+ p^\star) =c_s/c $, and thus $L^\star=1+p^\star$ is the unique mixed strategy Nash equilibrium.
\begin{figure}
\centering
\post{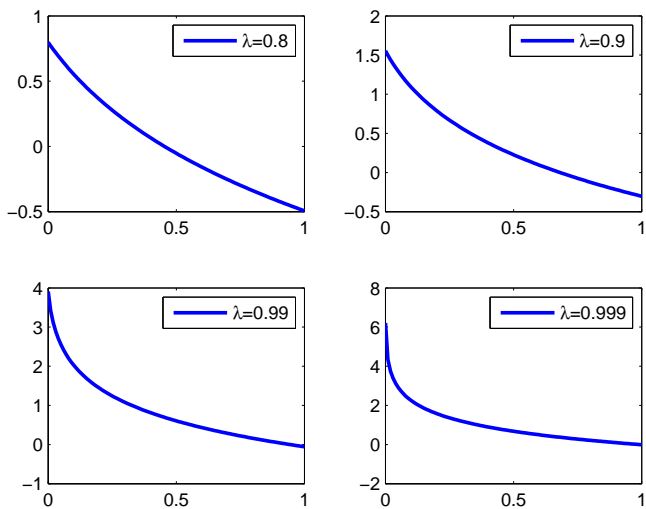}{4.0in}
\caption{The natural logarithm of $V(1,1+q)$ with increasing $q$ from $0$ to $1$, for varying arrival rate $\lambda$.}
\label{FigureValueofInformaiton}
\end{figure}

\subsection{Externality and Social Optimum}
The action of sampling more queues by some customers has an effect on the waiting time of others. This effect is called the {\it externality} associated with the action and the externality is positive if the action reduces the mean waiting time of the other customers. In this subsection, the externality in the mean field model is analyzed. It is not clear whether the externality is positive at first sight. On one hand, choosing a large number of queues to sample helps a customer find a less loaded queue and hence reduces future arrivals' opportunity to find lightly loaded queues. On the other hand, it also leads to a well balanced system and reduces the average waiting time.

The following corollary of Lemma \ref{LemmaMonotonicity} implies that the action of sampling more queues by some customers has a positive externality on the other customers in the mean field model. To see it, suppose in system 1, all the customers adopt a strategy $\mu_1$; while in system 2, a fraction $0<p\le1$ of them samples more queues, i.e., adopts a new strategy $\mu_3$ with $\mu_1 <_{\rm st} \mu_3$ and all the others still adopt the strategy $\mu_1$. For system $2$, it is equivalent to assume that all the customers adopt a strategy $\mu_2$ with $\mu_2 = p \mu_3+(1-p) \mu_1$. It follows that $\mu_1 <_{\rm st} \mu_2$. By Corollary \ref{CorollaryPositiveExternality}, system 2 has smaller mean waiting time.
\begin{corollary}\label{CorollaryPositiveExternality}
If $\mu_1, \mu_2 \in \mathcal{P}(\mathcal{L})$ with $\mu_1 <_{\rm st} \mu_2$, then for all $L \in \mathcal{L}$, $\mathbb{E} [W (L, \mu_1) ] > \mathbb{E} [W (L, \mu_2) ] $.
\end{corollary}
\begin{proof}
Because $\mu_1 <_{\rm st} \mu_2$, by Lemma \ref{LemmaMonotonicity}, $r_{\mu_1 }(k) > r_{\mu_2 }(k), \forall k \ge 2$. Hence, the conclusion follows by invoking the definition of $\mathbb{E} [W (L, \mu)]$.
\end{proof}
\begin{remark}
However, the action of sampling more queues by some customers can have a negative externality for finite $N$. See Section \ref{SectionExternalityFinite}.
\end{remark}
Next, we analyze the social optimum, i.e, the minimum of  total cost of all the customers. Suppose all the customers use the mixed strategy $\mu$; the expected total cost per unit time $C_{\rm sum} (\mu) $ is given by
\begin{align}
C_{\rm sum} ( \mu) =\lambda C(\mu,\mu)=  \lambda \sum_{L=1}^{L_{\max} } \mu(L) C(L,\mu) =  \lambda \sum_{L=1}^{L_{\max} } \mu(L) (  \sum_{k=0}^\infty r_{\mu}^L (k) +c_s L ). \nonumber
\end{align}
Therefore, a minimizer of the following minimization problem:
\begin{align}
\min_{\mu \in \mathcal{P}(L) }  \sum_{L=1}^{L_{\max} } \mu(L) (  \sum_{k=0}^\infty r_{\mu}^L (k) +c_s L ), \nonumber
\end{align}
is a social optimum.
\begin{lemma} \label{LemmaSocialOptimum}
The social optimum $\mu^\star_{\rm soc}$ is a probability measure concentrated on either an integer or two consecutive integers.
\end{lemma}
\begin{proof}
See the proof in Appendix \ref{ProofLemmaSocialOptimum}.
\end{proof}

\begin{theorem}
No Nash equilibrium $\mu^\star$ can strictly stochastically dominate the social optima $ \mu^\star_{\rm soc}$ in the mean field model.
\end{theorem}
\begin{proof}
The total cost can be decomposed into two terms as
\begin{align}
 \sum_{L=1}^{L_{\max} } \mu(L) C(L, \mu) =  \sum_{L=1}^{L_{\max} } \mu(L) C(L, \mu^\star) +  \sum_{L=1}^{L_{\max} } \mu(L) \left( C(L,\mu)- C(L,\mu^\star) \right) . \nonumber
\end{align}
Suppose the Nash equilibrium $\mu^\star >_{\rm st} \mu^\star_{\rm soc}$, then by the positive externality result, we have
\begin{align}
C (L,\mu^\star_{\rm soc} )- C(L,\mu^\star)   >0. \nonumber
\end{align}
Also, by definition of $\mu^\star$,
\begin{align}
\sum_{L=1}^{L_{\max} } \mu_{\rm soc}^\star (L) C(L, \mu^\star)   \ge \sum_{L=1}^{L_{\max} } \mu^\star (L) C(L, \mu^\star). \nonumber
\end{align}
Therefore,
\begin{align}
\sum_{L=1}^{L_{\max} } \mu^{\star}_{\rm soc}(L) C(L, \mu_{\rm soc}^\star )> \sum_{L=1}^{L_{\max} } \mu^\star(L) C(L, \mu^\star) , \nonumber
\end{align}
which is a contradiction to the definition of $\mu_{\rm soc}^\star$.
\end{proof}
\begin{remark}
Since Nash equilibrium $L^\star$ and social optimum $ \mu^\star_{\rm soc}$ can be identified with real numbers in $ [1,L_{\max}]$, the above lemma further implies that $L^\star \le L^\star_{\rm soc}$, i.e., no Nash equilibrium can be above the social optimum.
\end{remark}

\subsection{Heterogeneous Waiting Cost} \label{SectionHetero}
In this subsection, the heterogeneous waiting cost is considered. In particular, assume that there is a nondegenerate continuous probability density function for waiting cost $c$, i.e., $\int_{0}^{c_{\max}} f(c) d c =1$.

Fix a customer $i$, let $L_i(\cdot)$ denote a function from $[0, c_{\max}]$ to $\mathcal{L}$. We call $L_i(\cdot)$ the strategy of customer $i$. In particular, if customer $i$ has a waiting cost $c$, then she chooses $L_i(c)$ queues to sample uniformly at random. Now suppose all the other customers use $L_{-i}(\cdot )$,  the expected total cost for customer $i$ is given by
\begin{align}
C(L_i, L_{-i} )= c \mathbb{E} [ W(L_i(c), L_{-i}) ] + c_s L_i(c).
\end{align}
The goal of customer $i$ is to minimize the expected total cost by choosing the optimal $L_i(\cdot)$. Define the best response correspondence as $\text{BR}(L_{-i}(\cdot)):=\arg \min_{L_i(\cdot) } C(L_i, L_{-i})$.

If $L_i(\cdot ) \in \text{BR}(L_{-i}(\cdot))$, then $L_i(\cdot)$ must be a nondecreasing step function with respect to $c$ as depicted in Fig.~\ref{FigureBestRespnseHeteoCost}, which is proved in the following lemma.
\begin{lemma}\label{LemmaBestResponseHetero}
Suppose $L_i(\cdot ) \in \text{BR}(L_{-i}(\cdot))$, then $L_i(\cdot ) $ is a nondecreasing step function in $c$.
\end{lemma}
\begin{proof}
Suppose $ 0 \le c<\hat{c} \le c_{\max}$, since $L_i(\cdot ) \in \text{BR}(L_{-i}(\cdot))$, we have
\begin{align}
& c \mathbb{E} [ W(L_i(c), L_{-i}) ] + c_s L_i(c) \le c \mathbb{E} [ W(L_i(\hat{c} ), L_{-i}) ] + c_s L_i(\hat{c}), \nonumber \\
& \hat{c} \mathbb{E} [ W(L_i(\hat{c}), L_{-i}) ] + c_s L_i(\hat{c}) \le \hat{c} \mathbb{E} [ W(L_i(c), L_{-i}) ] + c_s L_i(c). \nonumber
\end{align}
Adding the above two inequalities up, we get
\begin{align}
(\hat{c}-c) ( \mathbb{E} [ W(L_i(\hat{c} ), L_{-i}) ] - \mathbb{E} [ W(L_i(c), L_{-i}) ] ) \le 0. \nonumber
\end{align}
Therefore $ \mathbb{E} [ W(L_i(\hat{c} ), L_{-i}) ] \le \mathbb{E} [ W(L_i(c), L_{-i}) ] $ and thus $L_i(\hat{c}) \ge L_i(c)$.
\end{proof}
\begin{figure}
\centering
\post{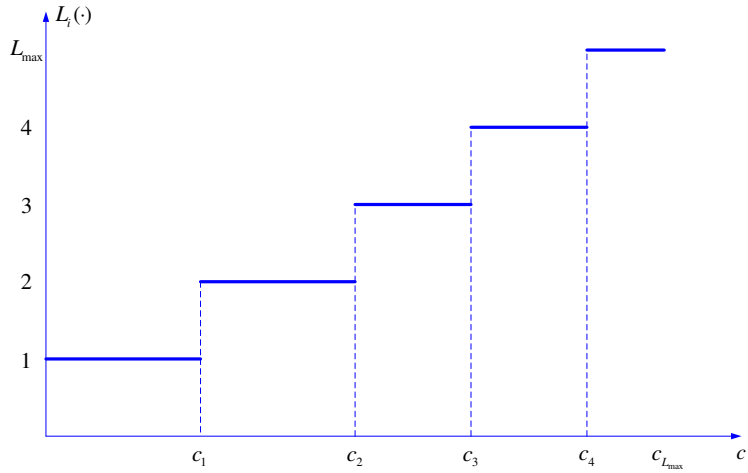}{4.0in}
\centering
\caption{The best response strategy $L_i(\cdot) \in \text{BR}(L_{-i}(\cdot))$ is a non-decreasing step function. The jumping points are denoted by $0=c_0 < c_1 < \cdots < c_{L_{\max}}=c_{\max}$.}
\label{FigureBestRespnseHeteoCost}
\end{figure}

Define the strategy space $\mathcal{S}$ as the collection of all possible nondecreasing step functions from $[0, c_{\max}]$ to $\mathcal{L}$. By Lemma \ref{LemmaBestResponseHetero}, it suffices to consider $\mathcal{S}$ for finding pure strategy Nash equilibrium. There is a bijective mapping between the strategy space $\mathcal{S}$ and probability space $\mathcal{P}( \mathcal{L})$. In particular, suppose $L(\cdot) \in \mathcal{S}$ is given, and let $0=c_0 < c_1 < \cdots < c_{L_{\max}}=c_{\max}$ denote the jumping points. Then, define $\mu(l) = \int_{c_{l-1}}^{c_l} f(c) d c$. On the other hand, suppose $\mu \in \mathcal{P}( \mathcal{L})$ is given, then the equation $\mu(l)= \int_{c_{l-1}}^{c_l} f(c) d c$ is solved to get the unique jumping points $0=c_0 < c_1 < \cdots < c_{L_{\max}}=c_{\max}$. Thus, $L(\cdot)$ can be constructed as $L(c)=l$ for $ c_{l-1} \le c < c_{l}$. Define a metric on the strategy space $\mathcal{S}$ as $ d (L_1 (\cdot ), L_2 (\cdot ) ) := \| \mu_{L_1} - \mu_{L_2} \|$. Also, Let $F$ denote the bijective mapping from $L(\cdot)$ to $\mu_{L}$ and $F^{-1}$ denote the inverse mapping. It is easy to see that $F$ and $F^{-1}$ are continuous.

Next, we show the existence of a pure strategy Nash equilibrium in the mean field model. First, let us derive the expression of $\mathbb{E} [ W(L_i, L_{-i}) ]$ in mean field equilibrium. Suppose all the customers except customer $i$ use the strategy $L_{-i}(\cdot) $. Due to the bijective mapping between the strategy space and probability space, it is equivalent to consider the case in which all the customers use the mixed strategy $\mu_{L_{-i}}$. Therefore, the mean field equilibrium distribution satisfies
$r_{L_{-i}} (k) = r_{\mu_{L_{-i}} } (k)$, and the expected waiting time of customer $i$ using strategy $L_i(\cdot)$ is given by
\begin{align}
\mathbb{E} [ W(L_i(c), L_{-i} ) ] = \mathbb{E} [ W (L_i(c), \mu_{L_{-i}} ) ] =  \sum_{k=0}^\infty  r_{\mu_{L_{-i}} }^{L_i(c)} (k). \nonumber
\end{align}

\begin{lemma} \label{LemmaContinuityHetero}
The best response correspondence $\text{BR}(L_{-i})$ is a continuous function.
\end{lemma}
\begin{proof}
See the proof in Appendix \ref{ProofLemmaContinuityHetero}.
\end{proof}

\begin{figure}
\centering
\post{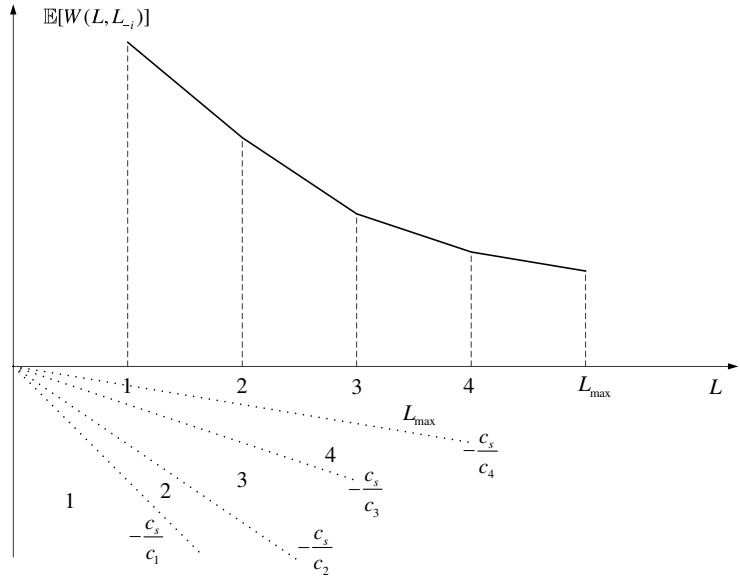}{4.0in}
\centering
\caption{The expected waiting cost of customer $i$ choosing $L$ while the others choosing a fixed $L_{-i}$. In this example, $L_{\max}=5$ and $c_1, \ldots, c_4$ are the jumping points of $L_i(\cdot) \in \text{BR}(L_{-i}(\cdot))$.}
\label{FigureWaitingCostHeteroCost}
\end{figure}

Define the best response from $\mathcal{P}(\mathcal{L})$ to $\mathcal{P}(\mathcal{L})$ as $ \tilde{\rm BR} (\mu_{-i}) = F( \text{BR}(F^{-1}(\mu_{-i})))$. Then the existence of a pure strategy Nash equilibrium is equivalent to the existence of a fixed point of $ \tilde{\rm BR} $. The Brouwer fixed point theorem is used to prove it.

\noindent {\bf (Brouwer's Theorem)} Every continuous function $g$ from a convex compact subset of a Euclidean space to itself has a fixed point.

\begin{theorem}
The supermarket game with heterogeneous waiting cost has a pure strategy Nash equilibrium in the mean field model.
\end{theorem}
\begin{proof}
 Note that $ \tilde{\rm BR}$ is a continuous function because $F$, $F^{-1}$, and $\text{BR}$ are continuous functions. Also, $\mathcal{P}(\mathcal{L})$  is a convex compact subset of an Euclidean space. Therefore, by Brouwer's fixed point theorem, $ \tilde{\rm BR}$ has a fixed point $\mu^\star$ and thus $F^{-1} (\mu^\star)$ is a pure strategy Nash equilibrium.
\end{proof}

In the sequel, the stochastic ordering between two possible pure strategy Nash equilibria is analyzed. For customer $i$, her marginal value of sampling  at an integer $L_i(c)$ with all the others adopting $L_{-i}(\cdot)$ is given by
\begin{align}
V(L_i(c), L_{-i}(\cdot)) = V(L_i(c), \mu_{L_{-i}})=\sum_{k=0}^\infty ( r_{\mu_{L_{-i}}}^{L_i(c)} (k) - r_{\mu_{L_{-i}}}^{L_i(c)+1 } (k)  ). \nonumber
\end{align}
The next lemma generalizes Lemma \ref{LemmaMonUnique} and proves a global monotonicity result of $V(L_i(c), L_{-i}(\cdot))$.

\begin{lemma} \label{LemmaMonUniqueHetero}
Assume $\lambda^2 \le 1/2 $ and fix any $L_{-i}(\cdot), \tilde{L}_{-i}(\cdot) \in \mathcal{S}$ such that $\mu_{L_{-i}} <_{\rm st} \mu_{\tilde{L}_{-i}}$. Then $ V(L_i(c), L_{-i} ) > V(L_i(c), \tilde{L}_{-i})$ for all $1 \le L_i(c) \le L_{\max}-1$. Furthermore, let $L_{i}(\cdot) \in \text{BR}( L_{-i}(\cdot)   )$ and $\tilde{L}_i(\cdot) \in \text{BR}(\tilde{L}_{-i}(\cdot) )$, then $\mu_{\tilde{L}_i } \le_{\rm st} \mu_{L_{i}}$.
\end{lemma}
\begin{proof}
By Lemma \ref{LemmaMonUnique}, it follows that $ V(L_i(c), L_{-i} ) > V(L_i(c), \tilde{L}_{-i})$ for all $1 \le L_i(c) \le L_{\max}-1$. Next, we show that $\mu_{\tilde{L}_i } \le_{\rm st} \mu_{L_{i}}$. Denote the jumping points of $L_{i}$ and $\tilde{L}_i$ by $0=c_0<c_1<\cdots<c_{L_{\max}}=c_{\max}$ and $0=\tilde{c}_0<\tilde{c}_{1}<\cdots<\tilde{c}_{L_{\max}}=c_{\max}$ respectively. It suffices to show that $c_{j} \le \tilde{c}_{j}$ for $0 \le j \le L_{\max}$. Fig.~\ref{FigureWaitingCostHeteroCost} shows that for $1 \le j \le L_{\max}-1$,
\begin{align}
c_j=\frac{c_s}{V(j,L_{-i} ) }, \quad \tilde{c}_j=\frac{c_s}{V(j,\tilde{L}_{-i} )}, \nonumber
\end{align}
which implies that $c_{j} \le \tilde{c}_{j}$.
\end{proof}

\begin{corollary}
Let $L_1^\star(\cdot)$ and $L_2^\star(\cdot) $ denote two possible distinct Nash equilibria for supermarket game with heterogeneous waiting cost in the mean field model. If $\lambda^2 \le \frac{1}{2}$, then $\mu_{L_1^\star}$ and $\mu_{L_2^\star}$ cannot be stochastically ordered.
\end{corollary}
\begin{proof}
Suppose $\mu_{L_1^\star} <_{\rm st} \mu_{L_2^\star}$. Then, by Lemma \ref{LemmaMonUniqueHetero}, $ \mu_{L_1^\star} \ge_{\rm st} \mu_{L_2^\star}$, which is a contradiction to assumption and concludes the proof.
\end{proof}
\begin{remark}
We are unable to prove the uniqueness of pure strategy Nash equilibrium because stochastic dominance is not a total order, i.e., there exists $\mu_1$ and $\mu_2$ which cannot be stochastically ordered.
\end{remark}

\section{Justification of Mean Field Model} \label{SectionFiniteN}
In this section, we justify the mean field model as the right limit of the supermarket game with finite $N$ as $N \to \infty$ by studying the equilibrium queue length distribution and $\epsilon$-Nash equilibrium.

\subsection{Propagation of Chaos}
In this subsection, we rigorously prove the propagation of chaos and coupling result for the finite $N$ model with all the customers using strategy $\mu \in \mathcal{P}(\mathcal{L})$. Our proof techniques are essentially the same as \citep{Graham00} and \citep{Turner98}. The only difference is that we consider a slightly more general case where customers' choices are random instead of deterministic and fixed.

Denote by $Q_i^N(t)$ the length of queue $i$ at time $t$. The process of $\{ Q_i^N \}$ is Markov, and the empirical distribution $\nu^N= (1/N) \sum_{i=1}^N \delta_{Q_i^N}$ has samples in $\mathcal{P}( \mathbb{D}( \mathbb{R}_+, \mathbb{N} ))$. Define the marginal process $\bar{Q}^N= (\bar{Q}^N_t)_{t\ge 0}$ as
\begin{align}
\bar{Q}^N_t = \nu_t^N = \frac{1}{N} \sum_{i=1}^N \delta_{Q_i^N(t)}. \nonumber
\end{align}
The marginal process $\bar{Q}^N$ has sample paths in $\mathbb{D}(\mathbb{R}_+, \mathcal{P}(\mathbb{N}) )$. The fraction of queues of length at least $k$ at time $t$ can be written as $r_t^N(k)=\bar{Q}^N_t([k,\infty))$.

For $ \Gamma \in \mathcal{P}(\mathscr{X}) $, a sequence of random variables $(X_i)_{1 \le i \le N}$ on $\mathscr{X}^N$ is $\Gamma$-{\it chaotic} if for any fixed integer $l \ge 1$, as $N \to \infty$,
\begin{align}
\mathscr{L} ( X_1, \ldots, X_l  ) & \Longrightarrow \Gamma^{ \otimes l}. \nonumber
\end{align}
A sequence of random variables $(X_i)_{1 \le i \le N}$ on $\mathscr{X}^N$ is {\it exchangeable} if for any permutation $\pi: [1,\ldots, N] \to [1,\ldots,N]$,
\begin{align}
\mathscr{L} (X_1, \ldots, X_N) = \mathscr{L} ( X_{\pi(1)}, \ldots, X_{\pi(N)} ). \nonumber
\end{align}

The proof roadmap is summarized in Fig.~\ref{FigureProofRoadMap}. We first prove a chaoticity result on path space (Thm. 7), i.e., there exists a $\Gamma \in \mathcal{P}( \mathbb{D}( \mathbb{R}_+, \mathbb{N} )) $ such that $(Q_i^N)_{ 1\le i \le N}$ is $\Gamma$-chaotic, if the initial condition $(Q_i^N(0))_{1 \le i \le N}$ is chaotic. Then, we show a chaoticity result in equilibrium (Thm. 10) by (i) taking the large $N$ limit and proving that the solution of the mean field equation converges as $t \to \infty$ to a fixed point (Lemma 13); (ii) taking the large $t$ limit and proving the finite $N$ model is ergodic (Thm. 9); (iii) using the chaoticity result on path space (Thm. 7) to finish the proof. The coupling result proved in Theorem 8 is used to show the ergodicity result in Theorem 9.
\begin{figure}
\centering
\post{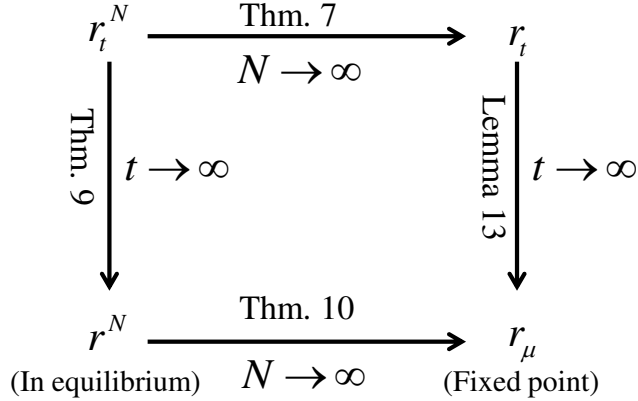}{3.5in}
\centering
\caption{The roadmap to show a chaoticity result in equilibrium.}
\label{FigureProofRoadMap}
\end{figure}
The chaoticity result on path space (Thm. 7) is proved using Proposition 2.2 in \citep{Sznitman91} and the {\it nonlinear martingale problem} approach. First, some useful definitions and preliminary lemmas are stated.

Let $(n)_k:=n(n-1)\cdots(n-k+1)$ for integer $n\ge k\ge 0$.  Define the $k$-body empirical measure as
\begin{align}
\nu^{k,N}= \frac{1}{(N)_k} \sum_{\substack{ i_1,\ldots, i_k=1, \\ \text{distinct} } }^N \delta_{Q_{i_1}^N, \ldots, Q_{i_k}^N} \nonumber
\end{align}
and the $k$-body empirical measure for queues other than $i$
\begin{align}
\nu_i^{k,N}= \frac{1}{(N-1)_k} \sum_{\substack{i_1,\ldots, i_k=1, \\ \text{distinct},\neq i}}^N \delta_{Q_{i_1}^N, \ldots, Q_{i_k}^N} \nonumber
\end{align}
with their marginal process $\bar{Q}^{k,N}$ and $\bar{Q}_i^{k,N}$.
The Lemma 3.1 in \citep{Graham00}, restated in the following lemma, proves that the $k$-body empirical measure is close to $k$-product of the empirical measure $\nu^N$.
\begin{lemma}\label{LemmaApprox}
$ \| \nu^{k,N}- (\nu^N)^{\otimes k} \|_{\rm TV} = O(1/N)$ and $\| \nu_i^{k,N} - (\nu^N)^{\otimes k} \|_{\rm TV} = O(1/N)$.
\end{lemma}
For a bounded function $\phi$ on $\mathbb{N}$, set $\phi^{+}(x)=\phi(x+1)-\phi(x)$ and $\phi^{-}(x)=\phi(x-1)-\phi(x)$. Let $\chi_L $ be a function as
\begin{align}
\chi_L : (x_1, \ldots, x_L) \in \mathbb{N}^L \mapsto \frac{\mathbf{1}_{x_1=\min \{x_1, \ldots, x_L \} } }{\sum_{i=1}^L \mathbf{1}_{x_i=\min \{x_1, \ldots, x_L \} }   } \in \{0, 1/L, \ldots, 1/2, 1 \}.  \nonumber
\end{align}
The next lemma gives a martingale process induced from the Markov process $\{ Q_i^N \}$.
\begin{lemma} \label{LemmaMartingalProblem}
$M^{\phi,i,N}$ is a martingale process defined by
\begin{align}
M_t^{\phi,i,N} = &  \phi(Q_i^N(t)) - \phi(Q_i^N(0)) - \int_{0}^t  \sum_{L=1}^{L_{\max}} \mu(L) \{ \lambda L \langle \chi_L ( Q_i^N(s), \cdot), \bar{Q}_i^{L-1,N} (s) \rangle  \nonumber \\
& \phi^{+}( Q_i^N(s) ) + \mathbf{1}_{Q_i^N(s) \ge 1} \phi^{-} (Q_i^N (s) ) \} ds \nonumber \\
= & \phi(Q_i^N(t)) - \phi(Q_i^N(0)) - \int_{0}^t  \sum_{L=1}^{L_{\max}} \mu(L) \{ \lambda L \langle \chi_L ( Q_i^N(s), \cdot), (\bar{Q}_s^N)^{\otimes L-1 } \rangle \nonumber \\
& \phi^{+}( Q_i^N(s) ) + \mathbf{1}_{Q_i^N(s) \ge 1} \phi^{-} (Q_i^N (s) ) \} ds + \epsilon^{\phi,i,N}(t) \nonumber
\end{align}
and $\epsilon^{\phi,i,N}(t)= t \| \phi \|_{\infty} O(\frac{1}{N})$ uniformly, $\langle M^{\phi,i,N}, M^{\phi,j,N} \rangle$ are zero for $i \neq j$ and $\mathbb{E} [\langle M^{\phi,i,N}, M^{\phi,i,N} \rangle_t ] \le C \| \phi \|^2_{\infty} t^2 $ for a constant $C$.
\end{lemma}
\begin{proof}
The martingale statement is obtained from the Dynkin formula in stochastic analysis. The bound on $\epsilon^{\phi,i,N}(t)$ follows from Lemma \ref{LemmaApprox}. Also, because there are no simultaneous jumps, the quadratic covariations $ \langle M^{\phi,i,N}, M^{\phi,j,N} \rangle $ are $0$ for $ i \neq j$. Lastly, by the definition of quadratic variation and $M_t^{\phi,i,N}$,
\begin{align}
\mathbb{E} [ \langle M^{\phi,i,N}, M^{\phi,i,N} \rangle_t ] &= \mathbb{E} [ ( M_t^{ \phi,i,N} )^2 ] \nonumber \\
&\le 3 \left( 2  \| \phi \|^2_{\infty} + \left( ( \lambda L_{\max} +1) \| \phi \|_{\infty} t  \right)^2 \right) \le C \| \phi \|^2_{\infty} t^2. \nonumber
\end{align}
\end{proof}

Next, we introduce the {\it nonlinear martingale problem} which is useful to prove the propagation of chaos result. We say a law $\Gamma \in \mathcal{P}( \mathbb{D}( \mathbb{R}_+, \mathbb{N} ))$ solves the nonlinear martingale problem if for any bounded function $\phi$ on $\mathbb{N}$,
\begin{align}
M_t^{\phi} &= \phi(Q_t) - \phi(Q_0) - \int_{0}^t \sum_{L=1}^{L_{\max}} \mu(L) \{ \lambda L \langle \chi_L ( Q_s, \cdot), {\Gamma}_s^{\otimes L-1 } \rangle \phi^{+}( Q_s ) \nonumber \\
& +  \mathbf{1}_{Q_s \ge 1} \phi^{-} (Q_s) \} ds \label{EqMartingaleProblem}
\end{align}
defines a $\Gamma$-martingale, where $Q \in \mathbb{D}( \mathbb{R}_+, \mathbb{N} )$ is distributed as $\Gamma$ and $\Gamma_t$ is the distribution of $Q_t$. It solves the martingale problem starting at $\gamma$ if furthermore $\Gamma_0=\gamma$.

By taking the expectation over two sides of (\ref{EqMartingaleProblem}) and using the result
\begin{align}
L\langle \chi_L(\mathbf{x} )\phi^{+}(x_1), {\Gamma}_s^{\otimes L } (d \mathbf{x} ) \rangle = \langle \phi^{+} (\min \{ \mathbf{x} \}), {\Gamma}_s^{\otimes L } (d \mathbf{x} ) \rangle, \nonumber
\end{align}
we get the {\it nonlinear Kolmogorov equation}. We say a (deterministic) process $(\Gamma_t)_{t \ge 0} \in \mathbb{D}( \mathbb{R}_+, \mathcal{P}( \mathbb{N} ) ) $ solves the nonlinear Kolmogorov equation if for any bounded function $\phi$ on $\mathbb{N}$,
\begin{align}
\langle \phi, \Gamma_t \rangle & = \langle \phi, \Gamma_0 \rangle + \int_0^t \sum_{L=1}^{L_{\max}} \mu(L) \{ \lambda \langle \phi^{+} (\min \{ \mathbf{x}\} ), \Gamma_s^{\otimes L}(d \mathbf{x}) \rangle \nonumber \\
& + \langle \mathbf{1}_{\cdot \ge 1} \phi^{-}, \Gamma_s \rangle \} ds \label{EqKolmogorov}
\end{align}
and it solves the equation starting at $\gamma$ if moreover $\Gamma_0=\gamma$. Taking $\phi$ equal to $\mathbf{1}_{[ k, \infty)}$, and using the fact that
\begin{align}
\langle \phi^{+} (\min \{ \mathbf{x} \} ), \Gamma_s^{\otimes L}(d \mathbf{x} ) \rangle= \sum_{x \in \mathbb{N}} \phi^{+}(x) \left( \Gamma_s^L([x, \infty) ) -\Gamma_s^L([x+1, \infty) ) \right) , \nonumber
\end{align}
we obtain the mean field equation for $r_t(k)=\Gamma_t([k,\infty))$ as
\begin{align}
\frac{d r_t(k) }{d t} = \lambda \sum_{L=1}^{L_{\rm max}}   \mu(L) \left( r_t^{L}(k-1) - r_t^{ L } (k) \right) - (r_t(k)-r_t(k+1)). \label{EqMeanFieldLimit}
\end{align}
The following theorem, which is essentially the same as Theorem 3.3 in \citep{Graham00}, proves that the nonlinear martingale problem has a unique solution.
\begin{theorem} \label{TheoremUniqMP}
Let $\gamma \in \mathcal{P}(\mathbb{N})$. Then there is a unique solution $\Gamma \in \mathcal{P}( \mathbb{D}( \mathbb{R}_+, \mathbb{N} ))$ for the nonlinear martingale problem (\ref{EqMartingaleProblem}) starting at $\gamma$. Its marginal process $(\Gamma_t)_{t\ge0}$ is the unique solution for the nonlinear Kolmogorov equation (\ref{EqKolmogorov}) starting at $\gamma$, and $r=(r_t)_{t \ge 0}$ defined by  $r_t(k)=\Gamma_t([k,\infty))$ is the unique solution for the mean field equation (\ref{EqMeanFieldLimit}) starting at $r_0$.
\end{theorem}
\begin{proof}
See the proof in Appendix \ref{ProofTheoremUniqMP}.
\end{proof}

The following theorem proves a propagation of chaos result on the path space.
\begin{theorem} \label{ThmChaoticityMeanFieldLimit}
Assume that $(\bar{Q}_0^N)_{N \ge 1}$ converges in law to a $\gamma$ in $\mathcal{P}(\mathbb{N})$, or that $(r_0^N)_{N \ge 1}$ converges in law to $r_0$ with $r_0(k)=\gamma([k,\infty))$. Then $(\nu^N)_{N \ge 1}$ converges in probability to the unique solution $\Gamma$ for the nonlinear martingale problem (\ref{EqMartingaleProblem}) starting at $\gamma$, and if $(Q_i^N(0))_{1 \le i \le N}$ is exchangeable then $(Q_i^N)_{ 1 \le i \le N} $ is $\Gamma$-chaotic. Moreover, $(\bar{Q}^N)_{N \ge 1}$ converges in probability to $(\Gamma_t)_{t \ge 0}$ and $(r^N)_{ N \ge 1}$ converges in probability to $(r_t)_{t \ge 0}$, for uniform convergence on bounded intervals over $t$, where $r_t(k)=\Gamma_t([k,\infty))$. Note that $(\Gamma_t)_{t \ge 0}$ is the unique solution for the Kolmogorov equation (\ref{EqKolmogorov}) starting at $\gamma$, and $(r_t)_{t \ge 0}$ is the unique solution for the mean field equation (\ref{EqMeanFieldLimit}) starting at $r_0$.
\end{theorem}
\begin{proof}
See the proof in Appendix \ref{ProofThmChaoticityMeanFieldLimit}.
\end{proof}

The following lemma proves that in the mean field model, starting with any appropriate initial distribution, the solution of mean field equation (\ref{EqMeanFieldLimit}) will converge to its fixed point.
\begin{lemma} \label{LemmaConvergenceMFE}
In the mean field model, starting with any initial distribution $r_0(k)$ such that $\sum_{k\ge 0} r_0(k) <\infty$, the solution of mean field equation (\ref{EqMeanFieldLimit}) $r_t(k)$  converges in $t$ to a fixed point $r(k)$ given by
\begin{align}
r(0)=1, r(k)=\lambda u_{\mu} (r(k-1)). \label{EqMeanFieldLimitSolution}
\end{align}
\end{lemma}
\begin{proof}
See the proof in Appendix \ref{ProofLemmaConvergenceMFE}.
\end{proof}

Let $r_t^{(N,j)} (k)$ and $Q_i^{(N,j)} (t) $ denote the fraction of queues with at least $k$ customers and the number of customers at queue $i$  respectively for system $j$ at time $t$. We can prove the following coupling result in the similar way as Turner did for the fixed number of sampling queues \citep{Turner98}.

\begin{theorem}\label{TheoremCoupling}
Suppose all the customers adopt $\mu_1$ in system 1 and $\mu_2$ in system 2 with $\mu_1 \le_{\rm st} \mu_2 $. Then there is a coupling between the two systems such that for all $t$ and $x$,
\begin{align}
\sum_{i} [ Q_i^{(N,2) }(t) - x ]_{+} \le \sum_{i} [ Q_i^{(N,1)}(t) - x ]_{+}, \nonumber
\end{align}
where $[a]_+ = \max \{a,0 \}$. It follows that for any nondecreasing convex function $h$ and for all $t$,
\begin{align}
\sum_{i} h( Q_i^{(N,2)}(t) ) \le  \sum_{i} h( Q_i^{(N,1)}(t) ). \nonumber
\end{align}
\end{theorem}
\begin{proof}
A coupling between $\mu_1$ and $\mu_2$ can be constructed as follows. let $F_1$ and $F_2$ be the cumulative distribution function for $\mu_1$ and $\mu_2$ respectively. Let $U$ denote the uniform distribution over $[0,1]$. Then it follows that $F_1^{-1} (U)$ and $F_2^{-1} (U) $ is distributed as $\mu_1$ and $\mu_2$ respectively. Also, $ F_1^{-1} (U) \le F_2^{-1} (U) $.

Let us couple the two systems as follows. All arrival times and departure times are the same in both systems, except that departures that occur from an empty queue are lost. At the times of departure, we arrange the queues in each system in order of queue lengths, and let a departure occur from the corresponding queue in each system; for example, if it occurs from the longest queue in system $1$, then let it also occur from the longest queue in system $2$. At arrival times, the customer in system $1$ samples $F_1^{-1} (U)$ queues; while the customer in system 2 first chooses the same set of queues chosen by system $1$ and then samples additional $( F_2^{-1} (U)-F_1^{-1} (U)) $ queues.

Using Theorem 4 in \citep{Turner98} and the fact that $ F_1^{-1} (U) \le F_2^{-1} (U)$, the conclusion follows.
\end{proof}

The coupling result is used to show the ergodicity for $Q_{i}^N$, which is proved in the following theorem.
\begin{theorem}\label{ThmErgodicity}
For $N \ge 1$, $(Q_i^N)_{1 \le i \le N}$ is ergodic for $\lambda<1$, and thus has a unique stationary distribution. Furthermore, in equilibrium $(Q_i^N)_{1 \le i \le N}$ is exchangeable.
\end{theorem}
\begin{proof}
The proof is the same as that in Theorem 4.2 in \citep{Graham00}.
\end{proof}

The following theorem proves the chaoticity in equilibrium.
\begin{theorem} \label{ThmChaoticityMFE}
Let $r_{\mu}$ be defined as (\ref{EqMeanFieldLimitSolution}) and define $\gamma_{\mu}\in \mathcal{P}(\mathbb{N}) $ as $\gamma_{\mu}(k)=r_{\mu}(k)-r_{\mu} (k+1)$. In equilibrium, $(Q_i^N)_{ 1 \le i \le N} $ is $\Gamma$-chaotic, where $\Gamma$ is the unique solution for the nonlinear martingale problem (\ref{EqMartingaleProblem}) starting at $\gamma_{\mu}$, and $Q$ (the Markov process with measure $\Gamma$) is in equilibrium under $\Gamma$. Hence in equilibrium $(\bar{Q}^N)_{N \ge1}$ converges in probability to the constant process identically equaling to $\gamma_{\mu}$ for all $t$ and $(r^N)_{N \ge 1}$ converges in probability to the constant process identically equaling to $r_{\mu}$ for all $t$, for uniform convergence on bounded intervals.
\end{theorem}
\begin{proof}
See the proof in Appendix \ref{ProofThmChaoticityMFE}.
\end{proof}

\begin{remark}
Let $Q_i^N (0) $ be the length of queue $i$ in equilibrium. Theorem \ref{ThmChaoticityMFE} implies that $(Q_i^N (0) )_{ 1 \le i \le N} $ is $\gamma_{\mu}$-chaotic. By the definition of chaoticity, it follows that the joint equilibrium distribution of any fixed number of queues converges to a product distribution, i.e., for any fixed integer $ l \ge 1$, as $N \to \infty$,
\begin{align}
\mathscr{L} ( Q_1^{(N)}(0), \ldots, Q_l^{(N)}(0)  ) & \Longrightarrow \gamma_{\mu}^{ \otimes l}. \nonumber
\end{align}
\end{remark}

\subsection{$\epsilon$-Nash Equilibrium for Finite $N$}
In this subsection, the $\epsilon$-Nash equilibrium for finite $N$ queues is considered. For a fixed customer $i$, suppose she uses the mixed strategy $\mu_i$, and all the other customers use $\mu_{-i}$. Let $W^{(N)}(\mu_i, \mu_{-i})$ and $C^{(N) }(\mu_i, \mu_{-i})$ denote her waiting time and total average cost for $N$ queues. Then, the expected waiting time can be derived as
\begin{align}
\mathbb{E} [W^{(N)}(L_i, \mu_{-i}) ] = \mathbb{E} \left[ \min \{ Q^{(N)}(1), \ldots, Q^{(N)}(L_i)  \} \right] +1, \nonumber
\end{align}
where $Q^{(N)}(i)$  is the length of the $i$th sampled queue.

Next, the definition of $\epsilon$-Nash equilibrium in the finite $N$ model is introduced.
\begin{definition}
We call $\mu^\star \in \mathcal{P}(\mathcal{L})$ an $\epsilon$-Nash equilibrium in the finite $N$ model, if for any $\epsilon >0$, with $N$ sufficiently large,
\begin{align}
C^{(N) }(\mu^\star, \mu^\star) \le C^{(N) }(\mu, \mu^\star )  +  \epsilon,  \text{ for all } \mu \in \mathcal{P}(\mathcal{L}).\nonumber
\end{align}
\end{definition}

\begin{theorem}
Let $\mu^\star$ be a Nash equilibrium for the supermarket game in the mean field model, then $\mu^\star$ is an $\epsilon$-Nash equilibrium in the finite $N$ model.
\end{theorem}
\begin{proof}
Let $Z$ denote the random variable associated with the mean field equilibrium distribution $r_{\mu_{-i}} (k)$. By Theorem \ref{ThmChaoticityMFE}, we have that for any fixed $L_i$, as $N \to \infty$,
\begin{align}
\mathscr{L} ( Q^{(N)}(1), \ldots, Q^{(N)}(L_i)  ) & \Longrightarrow \mathscr{L} (Z)^{ \otimes L_i} . \nonumber
\end{align}
By the coupling result in Theorem \ref{TheoremCoupling},
\begin{align}
\mathbb{E}_{\mu_{-i}} \left[ \min \{ Q^{(N)}(1), \ldots, Q^{(N)}(L_i)  \} \right] \le \mathbb{E}_{\mu_{-i}}  ( Q^{(N)}(1) ) \le \mathbb{E}_{\delta_{1}}  ( Q^{(N)}(1) )= \frac{\lambda}{1-\lambda},\nonumber
\end{align}
and hence $ W^{(N)} (L_i, \mu_{-i}) $ is uniformly integrable.
Therefore, for $\forall \epsilon >0$, there exists $N_0 \in \mathbb{N}$  such that when $N \ge N_0$,
\begin{align}
 | \mathbb{E} [ W^{(N)} (L_i, \mu_{-i}) ] - \mathbb{E} [ W (L_i, \mu_{-i}) ] | \le \frac{\epsilon}{2 c},\nonumber
\end{align}
and thus
\begin{align}
| C^{(N) } (\mu, \mu_{-i})- C(\mu, \mu_{-i}) | \le \epsilon /2 , \nonumber
\end{align}
where $W (L_i, \mu_{-i}) $ and $C(\mu, \mu_{-i})$ are the waiting time and total average cost respectively in the mean field model. By definition of $\mu^\star$,
\begin{align}
C(\mu^\star, \mu^\star) \le C(\mu, \mu^\star) ,  \text{ for all } \mu \in \mathcal{P}(\mathcal{L}). \nonumber
\end{align}
Then, it follows that
\begin{align}
 C^{(N) }(\mu^\star, \mu^\star) \le C^{(N) }(\mu, \mu^\star )  +  \epsilon,  \text{ for all } \mu \in \mathcal{P}(\mathcal{L}).\nonumber
\end{align}
Therefore, $\mu^\star$ is an $\epsilon$-Nash equilibrium in the finite $N$ system.
\end{proof}

\section{Externality for finite $N$} \label{SectionExternalityFinite}
In this section, we study the externality of sampling more queues by some customers in the finite $N$ model.

The following corollary shows that the action of sampling more queues by some customers has a nonnegative externality on the other customers who only sample one queue.

\begin{corollary}
If $\mu_1, \mu_2 \in \mathcal{P}(\mathcal{L})$ with $\mu_1 <_{\rm st} \mu_2$, then $\mathbb{E} [W (1, \mu_1) ] \ge \mathbb{E} [W (1, \mu_2) ] $.
\end{corollary}
\begin{proof}
Suppose all the customers adopt $\mu_1$ in system 1 and $\mu_2$ in system 2 with $\mu_1 \le_{\rm st} \mu_2$. Since $\mathbb{E} [W(1, \mu_i)]=\mathbb{E} [ \frac{1}{N}  \sum_{j}  Q_i^{(N)}(j) ] $ for $i=1,2$, it follows from Theorem \ref{TheoremCoupling} that $\mathbb{E}[W(1,\mu_1)] \ge \mathbb{E}[W(1, \mu_2)]$.
\end{proof}

For customers who sample more than one queue, the sampling of more queues by other customers has a negative externality in the following example.

\noindent {\bf Example}: Consider the supermarket game with two servers. All the customers adopt $\mu_1=\delta_{1}$ in system $1$ and  $\mu_2=\delta_{2}$ in system $2$.

For system $1$, queue 1 and queue 2 are two independent $M/M/1$ queues. Hence, the tail of the equilibrium queue length distribution is given by $r_{\mu_1} (k) = \lambda^k$ and thus $\mathbb{E} [W(2, \mu_1) ]= \sum_{k=0}^\infty r^2_{\mu_1} (k)= 1/(1-\lambda^2)$.

For system $2$. In the heavy traffic limit $\lambda \to 1$, Kingman showed that $ \mathbb{E} [W (2, \mu_2) ] \approx 1/(1-\lambda^2)$ in Theorem 6 in \citep{Kingman66}. In fact, $\mathbb{E} [W (2, \mu_2) ]> 1/(1-\lambda^2)$ for all arrival rates for the following reason. System $2$ is the same as the classical two parallel queues model with the routing-to-the-shortest-queue policy. A well known coupling argument implies that system $2$ has a strictly larger expected waiting time than the $M/M/2$ queueing model of two independent servers with unit exponential service rate and a single shared queueing buffer with total arrival rate $2\lambda$. The $M/M/2$ queueing model can be easily modeled as a birth-death process and the equilibrium queue length distribution is given by
\begin{align}
\pi(0)= (1-\lambda)/(1+\lambda), \; \pi(k)=2 \lambda^k (1-\lambda)/(1+\lambda), k \ge 1,
\end{align}
and thus the expected queue length is $\frac{2 \lambda}{1-\lambda^2}$. By Little's Law, the expected waiting time is $ \frac{1}{1-\lambda^2}$.

By comparing system $1$ with system $2$, it follows that $\mathbb{E} [W (2, \mu_2) ]> \mathbb{E} [W(2, \mu_1) ]$. Therefore, the sampling of more queues by some customers can have a negative externality on customers sampling more than one queue.

\section{Conclusions}\label{SectionConclusion}
Our results indicate that the equilibrium picture for the supermarket game can
be somewhat complicated, at least if $0.5 < \lambda^2 < 1$.  In particular,
there may be multiple Nash equilibria, stemming from the fact that
customers do not always have an ``avoid the crowd'' response as when $\lambda^2 \leq 0.5$.
However,  this complication seems to occur only for $\lambda$ close to one and
$L_{\max}$ fairly large.  Also, at least in the mean field model, the positive
externality of increased sampling holds for the whole range $\lambda < 1$. For the finite $N$ model, the coupling result in \citep{Turner98} shows that sampling more queues by some customers has a positive externality on customers who only sample one queue. However, for customers who sample more than one queue, the samplings of more queues by other customers can have a negative externality, as shown by the example in Section \ref{SectionExternalityFinite}.

\appendix
\section{ADDITIONAL PROOFS}

\subsection{Proof of Lemma \ref{LemmaContinuity} } \label{ProofLemmaContinuity}
First, let us show that $C(\mu_i,\mu_{-i})$ is continuous in $\mu_i$ for any fixed $\mu_{-i}$. Let $\| \mu_i^{(n)}-\mu_i \|_{\rm TV} \to 0$ as $n \to \infty$, then we have
\begin{align}
\mid C(\mu_i^{(n)},\mu_{-i}) -C(\mu_i,\mu_{-i}) \mid &= | \sum_{L_i=1}^{L_{\max}} C(L_i,\mu_{-i} ) \left( \mu_i^{(n)}(L_i)- \mu_i(L_i) \right) | \nonumber \\
 & \overset{(a)}{\le} K \sum_{L_i=1}^{L_{\max}}  \mid \mu^{(n)}(L_i)- \mu(L_i) \mid \to 0, \text{as } n \to \infty,\nonumber
\end{align}
where $(a)$ follows from the fact that $C(L_i,\mu_{-i})$ is uniformly bounded, which is proved in Lemma \ref{LemmaMonotonicity}. Furthermore, $\mathcal{P}(\mathcal{L})$ is a compact set and thus $C$ is uniformly continuous in $\mu_i$ for any fixed $\mu_{-i}$.

Second, we show that $C$ is continuous in $\mu_{-i}$ for any fixed $\mu_i$. It suffices to show that $\mathbb{E} [ W(L, \mu_{-i}) ]$ is continuous in  $\mu_{-i}$.
Let $\| \mu_{-i}^{(n)}-\mu_{-i} \|_{\rm TV} \to 0$ as $n \to \infty$. First, note that the mean field equilibrium distribution $r_{\mu_{-i}}(k)$ is continuous in $\mu_{-i}$ for each $k$. Second, by Lemma \ref{LemmaMonotonicity}, $r_{\mu} (k) \le \lambda^k, \forall k, \forall \mu \in \mathcal{P}(\mathcal{L})$. Therefore, for $\forall \epsilon >0$, by choosing sufficiently large $K$,
\begin{align}
| \mathbb{E} [ W(L, \mu_{-i}^{(n)}) ] -\sum_{k=0}^K  r_{\mu_{-i}^{(n)}}^L(k) | \le \frac{\epsilon}{3} \nonumber \\
| \mathbb{E} [ W(L, \mu_{-i} ) ] -\sum_{k=0}^K  r_{\mu_{-i}}^L(k) | \le \frac{\epsilon}{3}.\nonumber
\end{align}
Also, since $r_{\mu_{-i}}(k)$ is continuous in $\mu_{-i}$ for each $k$, by choosing sufficiently large $N_0$, it follows that when $n \ge N_0$,
\begin{align}
|\sum_{k=0}^K  r_{\mu_{-i}^{(n)}}^L(k)- \sum_{k=0}^K  r_{\mu_{-i}}^L(k) | \le \frac{\epsilon}{3}.\nonumber
\end{align}
Therefore,
\begin{align}
| \mathbb{E} [ W(L, \mu_{-i}^{(n)}) ] - \mathbb{E} [ W(L, \mu_{-i}) ] |  \le \epsilon.\nonumber
\end{align}
Furthermore, $\mathcal{P}(\mathcal{L})$ is a compact set and thus $C$ is uniformly continuous in $\mu_{-i}$ for any fixed $\mu_i$.

Lastly, let us show that $C(\mu_i,\mu_{-i})$ is jointly continuous with respect to $\mu_i$ and $\mu_{-i}$. Let $(\mu_i^{(n)}, \mu_{-i}^{(n)}) \to (\mu_i, \mu_{-i})$ in total variation distance as $n \to \infty$, then
\begin{align}
& \mid C(\mu_i^{(n)},\mu_{-i}^{(n)}) - C (\mu_i, \mu_{-i}) \mid \nonumber \\
& \le \mid C(\mu_i^{(n)}, \mu_{-i}^{(n)}) - C( \mu_i, \mu_{-i}^{(n)} ) \mid + \mid C(\mu_i, \mu_{-i}^{(n)})- C( \mu_i, \mu_{-i}) \mid \nonumber \\
&= \mid \sum_{L_i=1}^{L_{\max}} C(L_i,\mu_{-i}^{(n)}) \left( \mu_i^{(n)}(L_i)- \mu_i(L_i) \right) \mid + \mid C(\mu_i, \mu_{-i}^{(n)})- C( \mu_i, \mu_{-i})\mid \nonumber \\
& \le K \sum_{L_i=1}^{L_{\max}} \mid  \mu_i^{(n)}(L_i)- \mu_i(L_i)  \mid + \mid C(\mu_i, \mu_{-i}^{(n)})- C( \mu_i, \mu_{-i})\mid  \nonumber \\
& \to 0, \text{as } n \to \infty,\nonumber
\end{align}
which concludes the proof.
\subsection{Proof of Theorem \ref{TheoremExistenceNE} } \label{ProofTheoremExistenceNE}
Since $\mu_i^{(n)} \in \text{BR}(\mu_{-i}^{(n)})$, for all $\hat{\mu}_i \in \mathcal{P}(\mathcal{L})$,
\begin{align}
C( \hat{\mu}_i, \mu_{-i}^{(n)}) \ge C( \mu_i^{(n)}, \mu_{-i}^{(n)}).
\end{align}
By the continuity of $C$ proved in Lemma \ref{LemmaContinuity},
\begin{align}
\lim_{n \to \infty} C( \hat{\mu}_i, \mu_{-i}^{(n)}) = C( \hat{\mu}_i, \mu_{-i}), \nonumber \\
\lim_{n \to \infty} C( \mu_i^{(n)}, \mu_{-i}^{(n)}) = C( \mu_i, \mu_{-i}). \nonumber
\end{align}
Therefore, for all $\hat{\mu}_i \in \mathcal{P}(\mathcal{L})$, $C( \hat{\mu}_i, \mu_{-i}) \ge C( \mu_i, \mu_{-i})$ and thus $\mu_i \in \text{BR} (\mu_{-i})$.
Then, by the Kakutani's Theorem, there must exist a $\mu^\star$ as a fixed point of $\text{BR}$ and thus $\mu^\star$ is a mixed strategy Nash equilibrium.
Furthermore, by Lemma \ref{LemmaBestResponse}, $\mu^\star$ can be identified with a real number $L^\star \in [1,L_{\max}]$.

\subsection{Proof of Lemma \ref{LemmaBRMarginalValueInfo} } \label{ProofLemmaBRMarginalValueInfo}
The first half of the lemma is proved first.
For any integer $L_i \in \mathcal{L}$ such that $L_i \in \text{BR} (\mu_{-i})$,
$C(L_i,\mu_{-i}) \le C(L_i-1,\mu_{-i})$ and $C(L_i,\mu_{-i}) \le C(L_i+1,\mu_{-i})$.
Thus equation (\ref{EqMarginalValueInfo}) easily follows by invoking the definition of $C(L_i,\mu_{-i})$.

Now suppose equation (\ref{EqMarginalValueInfo}) holds.
Since $V(L_i, \mu_{-i} )$ is strictly decreasing in $L_i$,
\begin{align}
 V(L, \mu_{-i})
\left\{
\begin{array}{rl}
\ge c_s/c  & \text{if } L < L_i ,\\
\le c_s/c  & \text{if } L > L_i, \nonumber
\end{array} \right.
\end{align}
which implies that for all $L \in \mathcal{L}$, $C(L_i,\mu_{-i}) \le C(L,\mu_{-i})$ and
thus $L_i \in \text{BR} (\mu_{-i})$.

The second half of the lemma is proved next. For any non-integer $L_i \in [1, L_{\max}]$ such that $L_i \in \text{BR} (\mu_{-i})$, $\lfloor L_i \rfloor \in \text{BR} (\mu_{-i})$ and $\lfloor L_i \rfloor+1 \in \text{BR} (\mu_{-i})$. It follows from equation (\ref{EqMarginalValueInfo}) that
\begin{align}
 V(\lfloor L_i \rfloor, \mu_{-i}) \le c_s/c \le  V(\lfloor L_i \rfloor +1 -1, \mu_{-i}),
\end{align}
which implies equation(\ref{EqMarginalValueInfoSecond}).

Now suppose equation (\ref{EqMarginalValueInfoSecond}) holds. Since $V(L_i,\mu_{-i})$ is strictly
decreasing in $L_i$,
\begin{align}
V (\lfloor L_i \rfloor , \mu_{-i} ) = c_s/c < V (\lfloor L_i \rfloor -1 , \mu_{-i} ),
\end{align}
which implies that $\lfloor L_i \rfloor \in \text{BR} (\mu_{-i})$ by the first half of the lemma just proved. Similarly, we can show that $\lfloor L_i \rfloor +1 \in \text{BR} (\mu_{-i})$. Therefore, $L_i \in \text{BR} (\mu_{-i})$.

\subsection{Proof of Lemma \ref{LemmaMonUnique}} \label{ProofLemmaMonUnique}
We first show that $ V(L_i, \mu_{-i}) > V(L_i, \tilde{\mu}_{-i} )$ for all $1 \le L_i \le L_{\max}-1$, i.e.,
\begin{align}
\sum_{k=0}^\infty ( r_{\mu_{-i}}^{L_i} (k) - r_{\mu_{-i}}^{L_i+1 } (k)  ) >  \sum_{k=0}^\infty ( r_{\tilde{\mu}_{-i}}^{L_i} (k) - r_{\tilde{\mu}_{-i}}^{L_i+1 }(k) ). \nonumber
\end{align}
We prove a stronger conclusion, that is, for $ k \ge 2$,
\begin{align}
r_{\mu_{-i}}^{L_i}  (k) - r_{\mu_{-i}}^{L_i+1 } (k)   > r_{\tilde{\mu}_{-i}}^{L_i} (k) - r_{\tilde{\mu}_{-i}}^{L_i+1 }(k). \nonumber
\end{align}
By Lemma \ref{LemmaMonotonicity}, for $k \ge 2$,
\begin{align}
r_{\tilde{\mu}_{-i}} (k) < r_{\mu_{-i}} (k) \le \lambda^k \le \frac{1}{2}. \nonumber
\end{align}
Define function $g(x)=x^{L_i}-x^{L_i+1}$. It is easy to calculate that $g^\prime (x)=x^{L_i-1}(L_i-(L_i+1)x )$. It follows that for $x < \frac{1}{2}$, $g^\prime (x)>0$. Therefore, $g( r_{\mu_{-i}} (k)) > g( r_{\tilde{\mu}_{-i}} (k)), \forall k \ge 2$.

Next, we show that $ \tilde{L}_i \le L_i $.

Case 1 ($\tilde{L}_i$ is an integer): Since $ \tilde{L}_i \in \text{BR}( \tilde{\mu}_i )$, it follows from Lemma \ref{LemmaBRMarginalValueInfo} that
\begin{align}
V(\tilde{L}_i-1, \mu_{-i}) > V(\tilde{L}_i-1, \tilde{\mu}_{-i} ) \ge c_s/c, \nonumber
\end{align}
which implies that $\lfloor L_i \rfloor > \tilde{L}_i-1$ again by Lemma \ref{LemmaBRMarginalValueInfo} and thus $L_i \ge \tilde{L}_i $.

Case 2 ($\tilde{L}_i$ is a non-integer): From Lemma \ref{LemmaBRMarginalValueInfo},
\begin{align}
V(\lfloor \tilde{L}_i \rfloor, \mu_{-i}) > V(\lfloor \tilde{L}_i \rfloor, \tilde{\mu}_{-i} ) = c_s/c. \nonumber
\end{align}
By Lemma \ref{LemmaBRMarginalValueInfo}, it follows that: (i) if $L_i$ is a non-integer, then $\lfloor L_i \rfloor > \lfloor \tilde{L}_i \rfloor$; (ii) if $L_i$ is an integer, then  $L_i > \lfloor \tilde{L}_i \rfloor $. Therefore, $L_i > \tilde{L}_i $.

\subsection{Proof of Lemma \ref{LemmaCharacterizationNE} } \label{ProofLemmaCharacterizationNE}
For case (a), $V(L_{\max}-1, L_{\max}) \ge \frac{c_s}{c}$. By Lemma \ref{LemmaBRMarginalValueInfo}, it follows that $L_{\max} \in \text{BR}(L_{\max})$ and thus $L_{\max}$ is a Nash equilibrium.

As for case (b), note that $\hat{L}$ is well defined, because $V(L_{\max}-1, L_{\max}) <  \frac{c_s}{c}$. In the case (b1), $V(\hat{L},\hat{L}) \le \frac{c_s}{c}$ and by the definition of $\hat{L}$, $V(\hat{L}-1,\hat{L}) \ge \frac{c_s}{c}$. Therefore, by Lemma \ref{LemmaBRMarginalValueInfo}, $\hat{L} \in \text{BR}(\hat{L})$. In the case (b2), $V(\hat{L},\hat{L}) > \frac{c_s}{c} $ and by the definition of $\hat{L}$, $V(\hat{L},\hat{L}+1 ) < \frac{c_s}{c} $. Since $V(L_i,L_{-i})$ is continuous in $L_{-i} \in \mathbb{R}$, it follows that there exists a $q^\star$ with $0<q^\star <1$ such that $V(\hat{L},\hat{L}+q^\star)=\frac{c_s}{c}$. Therefore, by Lemma \ref{LemmaBRMarginalValueInfo},  $\hat{L}+q^\star \in \text{BR}(\hat{L}+q^\star)$ and thus $\hat{L}+q^\star$ is a Nash equilibrium.

\subsection{Proof of Theorem \ref{TheoremUniqNE}} \label{ProofTheoremUniqNE}
\noindent (sufficiency) Let $L^\star \in [1, L_{\max}]$ denote any Nash equilibrium, and define $\hat{L}:= \min \{ L \in \mathcal{L} : V(L, L+1) < c_s/c \}$. By definition, $V(L_{\max},L_{\max} +1 ) = 0$ and thus $\hat{L}$ is well defined. Also, it follows that\begin{align}
V(L,L+1)>V(L+1,L+1)>V(L+1,L+2), \text{ for } 1 \le L \le L_{\max}-2, \nonumber
\end{align}
where the first inequality follows from the fact that $V(L_i, \mu_{-i})$ is strictly decreasing in $L_i$, and the second inequality follows from the monotonicity assumption of $V(L, L+q)$ in $q$. Therefore, $V(L,L+1)$ is strictly decreasing in $L$.

Next, the following three different cases of $\hat{L}$ are considered:

Case 1 ($\hat{L}= L_{\max}$): By the definition of $\hat{L}$, it follows that $V(L,L+1) \ge c_s/c$ for all $ 1 \le L \le L_{\max}-1$. Hence,
\begin{align}
V(L, L) > V(L, L+1) \ge c_s/c, \text{ for }  1 \le L \le L_{\max}-1. \label{EqMarginalValueMonotone}
\end{align}
Therefore, by Lemma \ref{LemmaBRMarginalValueInfo}, $L^\star$ cannot be an integer less than $L_{\max}$.
Now if $L^\star$ is a non-integer less than $L_{\max}$, then $\lfloor L^\star \rfloor \le L_{\max}-1$. Thus, it follows from (\ref{EqMarginalValueMonotone}) that
\begin{align}
V(\lfloor L^\star \rfloor, L^\star) > V(\lfloor L^\star \rfloor, \lfloor L^\star \rfloor +1) \ge c_s/c. \label{EqNEUniqValueInfo}
\end{align}
On the other hand, by the definition of $L^\star$ and Lemma \ref{LemmaBRMarginalValueInfo}, $V(\lfloor L^\star \rfloor, L^\star) = c_s/c$, which is a contradiction to (\ref{EqNEUniqValueInfo}). Therefore, $L^\star$ must be $L_{\max}$.

Case 2 ($\hat{L} \le L_{\max}- 1 $): Since $V(L,L+1)$ is strictly decreasing in $L$, by the definition of $\hat{L}$, $V(L_{\max}-1, L_{\max}) <  c_s/c$. Therefore, by Lemma \ref{LemmaBRMarginalValueInfo}, $L^\star$ cannot be $L_{\max}$.

Case 2(a): If $L^\star$ is an integer less than $L_{\max}$, then
\begin{align}
V(L^\star, L^\star) \le c_s/c \le V(L^\star-1, L^\star). \nonumber
\end{align}
By monotonicity assumption, it follows that
\begin{align}
&V(L, L+1) \ge V(L^\star-1, L^\star) \ge  c_s/c \text{ for } L \le L^\star-1, \nonumber \\
 &V(L^\star, L^\star +1) < V(L^\star, L^\star) \le c_s/c. \nonumber
\end{align}
Therefore, $L^\star -1 < \hat{L} \le L^\star $ and thus $L^\star =\hat{L}$.

Case 2(b): If $L^\star$ is a non-integer, then $V(\lfloor L^\star \rfloor, L^\star)=c_s/c$. By monotonicity assumption,
\begin{align}
V(\lfloor L^\star \rfloor, \lfloor L^\star \rfloor+1) < V(\lfloor L^\star \rfloor, L^\star) =\frac{c_s}{c}.
\end{align}
Also, for $L < \lfloor L^\star \rfloor -1$ ,
\begin{align}
V(L, L+1) > V(\lfloor L^\star \rfloor -1 , \lfloor L^\star \rfloor) > V(\lfloor L^\star \rfloor, \lfloor L^\star \rfloor) > V(\lfloor L^\star \rfloor, L^\star)=c_s/c. \nonumber
\end{align}
Therefore, $ \lfloor L^\star \rfloor -1< \hat{L} \le  \lfloor L^\star \rfloor $ and thus $\lfloor L^\star \rfloor =\hat{L}$. Furthermore, since $ V(\lfloor L^\star \rfloor, \lfloor L^\star \rfloor+1) < c_s/c < V(\lfloor L^\star \rfloor, \lfloor L^\star \rfloor)$, by monotonicity assumption, there is a unique $q^\star$ with $0<q^\star<1$ such that $V(\hat{L}, \hat{L}+q^\star)=\frac{c_s}{c}$ and thus $L^\star =\hat{L}+q^\star$.

\noindent (Necessary) Suppose the local monotonicity condition is not satisfied. It is well known that a continuous, injective (i.e. one-to-one) function $g$ on an interval $[a,b]$ is
either strictly monotone increasing or strictly monotone decreasing, so there are two cases to consider.

Case 1: There exists a $L$ with $1 \le L \le L_{\max}$ such that $V(L,L+q)$ is strictly increasing for $0\le q \le1$. Because $V(L-1,L)>V(L,L)$,  $c$ and $c_s$ can be selected to satisfy
\begin{align}
V(L,L) \le c_s/c = V(L,L+q_0) \le V(L-1,L) \nonumber
\end{align}
for some $0<q_0<1$. From Lemma \ref{LemmaBRMarginalValueInfo}, it follows that $L$ is a pure strategy Nash equilibrium and $L+q_0$ is a mixed strategy Nash equilibrium. Therefore, the supermarket game has at least two Nash equilibria.

Case 2: There exists a $L$ with $1 \le L \le L_{\max}$ such that $V(L,L+q_2)=V(L,L+q_3)$ with $0 \leq  q_2 < q_3 \leq 1$. Then, if $c$ and $c_s$ are selected so that
$\frac{c_s}{c}=V(L,L+q_2)=V(L,L+q_3)$,  both $L+q_2$ and $L+q_3$ are Nash equilibria.

Therefore, in either case, there are at least two Nash equilibria and thus the local monotonicity condition is necessary.

\subsection{Proof of Lemma \ref{LemmaSocialOptimum} } \label{ProofLemmaSocialOptimum}
Suppose $\mu \in \mathcal{P}(\mathcal{L})$ such that there exists $L_1 < L_2 < L_3$ with $\mu ( L_1) > 0$ and $\mu ( L_3) > 0$. Choose $ 0<\alpha, \beta \le \mu(L_1)$ such that $\alpha L_1 + \beta L_3 =(\alpha+\beta) L_2$. Then, we construct a new mixed strategy $\tilde{\mu}$ as
\begin{align}
\tilde{\mu}(L_1) = \mu ( L_1) -\alpha, \tilde{\mu}(L_3) = \mu ( L_3) -\beta ,\tilde{\mu}(L_2) = \mu ( L_2) +\alpha+ \beta, \nonumber
\end{align}
and $\tilde{\mu}(L)=\mu(L)$ for $L \neq  L_1,L_2,L_3$.
By the definition of $u_{\mu} (x)$ and convexity,
\begin{align}
u_{\mu} (x)- u_{\tilde{\mu} } (x) = \alpha x^{L_1} + \beta x^{L_3} -(\alpha+\beta) x^{L_2} >0. \nonumber
\end{align}
Therefore, $ r_{\tilde{\mu}} (k) <  r_{\mu} (k), k \ge 2 $. It follows then
\begin{align}
C_{\rm sum}(\tilde{\mu})  < \lambda \sum_{L=1}^{L_{\max} } \tilde{\mu}(L) \left(  \sum_{k=0}^\infty r_{\mu}^L (k) +c_s L \right) < C_{\rm sum} ( \mu ), \nonumber
\end{align}
where the second inequality follows from the strict convexity of $\mathbb{E}[W(L, \mu)]$ in $L$. Therefore, $\mu \neq \mu^\star_{\rm soc}$.

\subsection{Proof of Lemma \ref{LemmaContinuityHetero} } \label{ProofLemmaContinuityHetero}
Suppose $L_{-i}^{(n)} \to L_{-i}$ as $n \to \infty$. By the definition of metric, we have $ \mu_{L_{-i}^{(n)}} \to \mu_{ L_{-i}}$. Since $\mathbb{E}[ W(L_i(c), L_{-i}) ] = \mathbb{E}[ W(L_i(c), \mu_{ L_{-i}} ) ]$, by the continuity of $\mathbb{E}[ W(L_i(c), \mu_{ L_{-i}}) ]$ with respect to $\mu_{ L_{-i}}$, we conclude that $\mathbb{E}[ W(L_i(c), L_{-i}) ]$ is continuous in $L_{-i} (\cdot) $.

Next, we show that the correspondence $\text{BR}(L_{-i})$ is in fact a function. Let $L_i \in \text{BR}(L_{-i})$ and $ 0=c_0< c_1<\cdots< c_{L_{\max}}=c_{\max} $ denote the jumping points of $L_i$. It follows from Fig.~\ref{FigureWaitingCostHeteroCost} that for $j=1, \ldots, L_{\max}-1$, $c_j$ is uniquely determined by
\begin{align}
c_j &= \frac{c_s}{ \mathbb{E}[ W(j, L_{-i}) ] - \mathbb{E}[ W(j+1, L_{-i}) ] }. \nonumber
\end{align}
Therefore, $L_i$ is unique and $\text{BR}(L_{-i})$ is a function from $\mathcal{S}$ to $\mathcal{S}$.

Finally, we show the continuity of $\text{BR}(L_{-i})$. Suppose $L_{-i}^{(n)} \to L_{-i}$, $L_{i}^{(n)} = \text{BR} (L_{-i}^{(n)} )$, and $L_{i}=\text{BR} (L_{-i})$, we prove that $L_{i}^{(n)} \to L_{i}$. Denote the jumping points of $L_{i}^{(n)} $ and $L_{i}$ by $ 0=c^{(n)} _0< c^{(n)} _1<\cdots< c^{(n)} _{L_{\max}}=c_{\max} $ and $ 0=c_0< c_1<\cdots< c_{L_{\max}}=c_{\max} $ respectively. It suffices to show that $c^{(n)}_j \to c_j$, for $j=0,\ldots, L_{\max}$. From Fig.~\ref{FigureWaitingCostHeteroCost}, we can see that for $j=1, \ldots, L_{\max}-1$,
\begin{align}
c^{(n)}_j  = \frac{c_s}{ \mathbb{E}[ W(j, L^{(n)}_{-i}) ] - \mathbb{E}[ W( j+1, L^{(n)}_{-i}) ] }. \nonumber
\end{align}
By the continuity of $\mathbb{E}[ W(L, L_{-i}) ]$ with respect to $L_{-i} (\cdot) $,
\begin{align}
c^{(n)}_j \to c_j, \text{ for } 1\le j\le L_{\max}-1, \nonumber
\end{align}
which further implies that $\mu_{L_{i}^{(n)}} \to \mu_{L_{i}}$ and concludes the proof.
\subsection{Proof of Theorem \ref{TheoremUniqMP} } \label{ProofTheoremUniqMP}
Define the jumping measure as
\begin{align}
J(\xi, x, dy)= \lambda \sum_{L=1}^{L_{\max}} \mu(L) L \langle \chi_L (x, \cdot), \xi^{\otimes L-1 } \rangle \delta_{x+1} (dy) + \mathbf{1}_{x\ge1} \delta_{x-1} (dy). \nonumber
\end{align}
It follows that $\| J(\xi,x)  \|_{\rm TV} \le \lambda L_{\max} + 1$. Also,
\begin{align}
 \| J(\xi, x)- J(\eta, x) \|_{\rm TV} & = \lambda \sum_{L=1}^{L_{\max}} \mu(L) L \mid  \langle \chi_L (x, \cdot), \xi^{\otimes L-1 } - \eta^{\otimes L-1 } \rangle \mid \nonumber \\
 &\le \lambda L_{\max} \| \xi^{\otimes L-1 } - \eta^{\otimes L-1 }  \|_{\rm TV}, \nonumber \\
  \| \xi^{\otimes L-1 } - \eta^{\otimes L-1 }  \|_{\rm TV} & \le \lambda L_{\max} (L_{\max}-1) \| \xi-\eta \|_{\rm TV}. \nonumber
\end{align}
Therefore, $\| J(\xi, x)- J(\eta, x) \|_{\rm TV}  \le \lambda L_{\max} (L_{\max}-1) \| \xi- \eta \|_{\rm TV}$, which concludes the proof using Proposition 2.3 in \citep{Graham00}.

\subsection{Proof of Theorem \ref{ThmChaoticityMeanFieldLimit}} \label{ProofThmChaoticityMeanFieldLimit}
It is proved using a method developed by Sznitman \citep{Sznitman91} with three essential steps: proving that
\begin{enumerate}
\item  $(\mathscr{L}(\nu^N))_{N\ge1}$ is tight in $\mathcal{P}(\mathcal{P}( \mathbb{D}( \mathbb{R}_+, \mathbb{N} )))$.
\item  the nonlinear martingale problem (\ref{EqMartingaleProblem}) is satisfied by any probability measure in the support of any accumulation point of $(\mathscr{L}(\nu^N))_{N\ge1}$.
\item the nonlinear martingale problem (\ref{EqMartingaleProblem})  has a unique solution $\Gamma$ starting at any $\gamma$ in $\mathcal{P}(\mathbb{N})$.
\end{enumerate}
Once these three steps are done, Step 2 and Step 3 imply that the only possible accumulation point of $(\mathscr{L}(\nu^N))_{N\ge1}$ is $\delta_{\Gamma}$, and Step 1 implies that $(\mathscr{L}(\nu^N))_{N\ge1}$ weakly converges to $\delta_{\Gamma}$. By Proposition 2.2 in \citep{Sznitman91}, the conclusions follows. The individual steps are proved as follows.

{\bf Step 1}. It is equivalent to prove that $(\mathscr{L}(Q_1^N))_{N\ge1}$ is tight in $\mathcal{P}( \mathbb{D}( \mathbb{R}_+, \mathbb{N} ))$, see Proposition 2.2 in \citep{Sznitman91}. The jumps of $Q_1^N$ are included in those of a Poisson process $A^N$ of rate $\lambda L_{\max} + 1$, and their size is bounded by $1$, so the modulus of continuity in $\mathbb{D}( \mathbb{R}_+, \mathbb{N} )$ of $Q_1^N$ is less than that of $A^N$. The law of $A^N$ does not depend on $ N\ge 1$, and since $(Q_1^N(0))_{N\ge1}$ converges and hence is tight, the basic tightness criterion in Ethier and Kurtz \citep{Kurtz86} holds.

{\bf Step 2}. Let $\Pi^N$ be the law of $\nu^N$ and $\Pi^\infty \in \mathcal{P}(\mathcal{P}( \mathbb{D}( \mathbb{R}_+, \mathbb{N} )))$ be an accumulation point. Take $0 \le s_1 < \cdots <s_k \le s < t$ and a bounded function $g$ on $\mathbb{N}^k$, and then define $G: \mathcal{P}( \mathbb{D}( \mathbb{R}_+, \mathbb{N} )) \mapsto \mathbb{R}$ as
\begin{align}
G (R) =\langle  \bigg( & \phi(Q_t) - \phi(Q_s) - \int_{s}^t \sum_{L=1}^{L_{\max}} \mu(L) \{ \lambda L \langle \chi_L ( Q_u, \cdot), {R}_u^{\otimes L-1 } \rangle \phi^{+}( Q_u ) \nonumber \\
 & +  \mathbf{1}_{Q_u \ge 1} \phi^{-} (Q_u) \} du  \bigg) g(Q_{s_1},\ldots, Q_{s_k}), R \rangle \nonumber
\end{align}
Set $g^{i,N}= g( Q_i^{(N)} (s_1),\ldots, Q_i^{(N)}(s_k) )$  and $Y_{s,t}=Y_t-Y_s$ for a process $Y$. It follows that
\begin{align}
& \langle G^2, \Pi^N \rangle \nonumber \\
& = \mathbb{E} \left[ G^2(\nu^N)\right] = \mathbb{E} \left[ ( \frac{1}{N} \sum_{i=1}^N (M_{s,t}^{\phi,i,N} - \epsilon_{s,t}^{\phi,i,N} ) g^{i,N} )^2 \right] \nonumber \\
&= \frac{1}{N} \mathbb{E} \left[  (M_{s,t}^{\phi,1,N} g^{1,N} )^2 \right] + \frac{N-1}{N} \mathbb{E} \left[  M_{s,t}^{\phi,1,N} M_{s,t}^{\phi,2,N} g^{1,N} g^{2,N}   \right] + O(\frac{1}{N}) \nonumber \\
&= \frac{ \| g \|^2_{\infty}}{N} \mathbb{E} [ \langle M^{\phi,1,N}, M^{\phi,1,N} \rangle_{s,t} ] + \frac{N-1}{N} \mathbb{E} [ \langle M^{\phi,1,N}, M^{\phi,2,N} \rangle_{s,t}   ] \| g \|^2_{\infty} + O(\frac{1}{N})\nonumber
\end{align}
and then Lemma \ref{LemmaMartingalProblem} implies that $\langle G^2, \Pi^N \rangle = O(1/N)$. The Fatou lemma implies that $\langle G^2, \Pi^\infty \rangle \le \lim_{N \to \infty} \langle G^2, \Pi^N \rangle =0$ and thus $G(R)=0$, $\Pi^\infty$-a.s. Since this holds for arbitrary bounded $g$ and $0 \le s_1 < \cdots <s_k \le s < t$, we conclude that $R$ solves the nonlinear martingale problem (\ref{EqMartingaleProblem}), $\Pi^\infty$-a.s. Furthermore, the continuity of $Q \to Q_0$ implies that $R_0=\gamma$, $\Pi^\infty$-a.s.

{\bf Step 3}. This result is proved in Theorem \ref{TheoremUniqMP}.

\subsection{Proof of Lemma \ref{LemmaConvergenceMFE} } \label{ProofLemmaConvergenceMFE}
First, note that increasing $r_0(k)$ only increases all $r_t(k)$, because $dr_t(k)/dt$ is non-decreasing in $r_t(j)$ for $j \neq k$ \citep{Deimling77}.
Therefore, it suffices to show the conclusion for the following two cases: $r_0(k) \ge r(k), \forall k$ and $r_0(k) \le r(k), \forall k$.

Second, define $v^K_t(k)=\sum_{j \ge k}^K r_t(j)$ and $v_t(k)=v^\infty_t(k)$. We  show that for the above two cases, $v_t(k)$ is uniformly bounded over all $t$ and $k$. If $r_0(k) \le r(k), \forall k$, since $r(k)$ is a fixed point, $r_t(k) \le r(k)$ for all $t$. Therefore,
\begin{align}
v_t(k) \le v_t(1) \le \sum_{j \ge 1} r(j) \le \sum_{j \ge 1} \lambda^j= \frac{\lambda}{1-\lambda}. \nonumber
\end{align}
If $r_0(k) \ge r(k), \forall k$, then $r_t(k) \ge r(k)$ for all $t$. From the mean field equation (\ref{EqMeanFieldLimit}),
\begin{align}
\frac{dv^K_t(1)}{dt} = \lambda - \lambda \sum_{L=1}^{L_{\rm max}}  \mu(L) r_t^{L}(K)  - ( r_t(1)-r_t(K+1) ) \le \lambda. \nonumber
\end{align}
It follows that  $v^K_t(1) \le \lambda t$ and thus $v_t(1) \le \lambda t$, which further implies that  $\lim_{K \to \infty} r_t(K) =0 $. Therefore, taking limit $K \to \infty$ over both sides of the above equation gives
\begin{align}
\frac{dv_t(1)}{dt} = \lambda -r_t(1) = r(1)-r_t(1) \le 0. \nonumber
\end{align}
Thus, $v_t(k) \le v_t(1) \le v_0(1)$.

Finally, in order to show that $\lim_{t\to \infty} r_t(k) =r(k)$, it suffices to show
\begin{align}
-\infty <\int_{0}^\infty (r_t(k)-r(k) ) < \infty, \nonumber
\end{align}
because under the above two cases, the sign of $(r_t(k)-r(k))$ does not change in $t$. We prove the second inequality by induction and the first inequality can be proved in the same way. The first inequality  trivially holds for $k=0$ since $r_t(0)=r(0)=1$. Now suppose it holds for $k-1$. From mean field equation (\ref{EqMeanFieldLimit}) and definition of $r(k)$,
\begin{align}
\frac{dv_t(k)}{dt} &= \lambda \sum_{L=1}^{L_{\rm max}} \mu(L) r_t^{L}(k-1) - r_t(k) \nonumber \\
&= \lambda \sum_{L=1}^{L_{\rm max}}   \mu(L) \left( r_t^{L}(k-1) - r^{ L } (k-1) \right) - (r_t(k)-r(k)). \nonumber
\end{align}
By integrating it, it follows that
\begin{align}
v_t(k)=v_0(k) + \int_{0}^t \left[ \lambda \sum_{L=1}^{L_{\rm max}}   \mu(L) \left( r_s^{L}(k-1) - r^{ L } (k-1) \right) - (r_s(k)-r(k)) \right] ds. \nonumber
\end{align}
By induction, $\int_{0}^\infty (r_t(k-1)-r(k-1))dt <\infty $. Therefore,
\begin{align}
\int_{0}^t   \sum_{L=1}^{L_{\rm max}}   \mu(L) \left( r_s^{L}(k-1) - r^{ L } (k-1) \right) ds \le L_{\rm max} \int_{0}^\infty (r_t(k-1)-r(k-1))dt. \nonumber
\end{align}
By assumption, $v_0(k)$ is bounded and we just proved that $v_t(k)$ is uniformly bounded in $t$. Thus, $\int_{0}^\infty (r_t(k)-r(k))dt <\infty $ and the conclusion follows.

\subsection{Proof of Theorem \ref{ThmChaoticityMFE}} \label{ProofThmChaoticityMFE}
Let $\mathscr{L}_{st} \in \mathcal{P}(\mathbb{D}( \mathbb{R}_+, \mathbb{N} ) )$ denote the stationary measure on the path space of the finite $N$ model. Let $\mathscr{L}_{st} (\bar{Q}_0^N) $ denote the distribution of $\bar{Q}_0^N$ under the stationary measure, and let  $\mathscr{L}_{st} ({Q}_1^N(0) )$ denote the distribution of ${Q}_1^N(0)$ under the stationary measure.

{\bf Step 1}. Prove that $(\mathscr{L}_{st} (\bar{Q}_0^N) )_{N\ge1}$ is tight. It is equivalent to prove that $(\mathscr{L}_{st} ({Q}_1^N(0) ))_{N\ge1}$ is tight. Let system 0 be the finite $N$ model with all the customers sampling one queue; while system $1$ be the finite $N$ model with all the customers using strategy $\mu$. The system 0 is an i.i.d. system. If both systems are started at the same initial value, with law being the stationary distribution of system 0. Then, by the coupling result in Theorem \ref{TheoremCoupling},
\begin{align}
\mathbb{E} [ Q_1^{N,1} (t) ] = \frac{1}{N} \mathbb{E} [\sum_{i=1}^N Q_i^{N,1} (t) ] \le \frac{1}{N} \mathbb{E} [\sum_{i=1}^N Q_i^{N,0} (t) ] =\frac{\lambda}{1-\lambda}. \nonumber
\end{align}
Since $\lambda <1$, by the ergodicity result in Theorem \ref{ThmErgodicity} and Fatou lemma,
\begin{align}
\mathbb{E}_{st} [Q_1^N(0)] \le \liminf_{t \to \infty} \mathbb{E}[  Q_1^{N,1} (t) ] \le \frac{\lambda}{1-\lambda}, \nonumber
\end{align}
which implies that for all $N \ge 1$, $\mathscr{L}_{st}(Q_1^N(0))( [K,\infty) ) \le \lambda (1-\lambda)^{-1} K^{-1}$ and hence $\mathscr{L}_{st} ({Q}_1^N(0) )_{N\ge1}$ is tight.

{\bf Step 2}. Prove that $ \mathscr{L}_{st} (\bar{Q}_0^N) $ weakly converges to $\delta_{\gamma_{\mu}}$. Let $\Pi_0^\infty$ be an accumulation point of $ \mathscr{L}_{st} (\bar{Q}_0^N) $. Suppose the finite $N$ model starts with the stationary distribution. Let $\Pi^N$ denote the law of $\nu^N$ and $\Pi^\infty$ denote an accumulation point of $\Pi^N$ .
Then, we can prove exactly as for Theorem \ref{ThmChaoticityMeanFieldLimit} that $\Pi^N$ is tight, and that any law $R$ in the support of $\Pi^\infty$ satisfies the non-linear martingale problem. In particular, $(R_t)_{t\ge 0}$ solves the nonlinear Kolmogorov equation. Since the initial distribution is the stationary distribution, $\Pi_0^N= \Pi_t^N $. Taking the limit along the subsequence, it follows that $ \Pi_0^\infty=\Pi_t^\infty$.

Take $\epsilon >0$ and an open neighborhood $V$ of $\gamma_{\mu}$. For $j \in \mathbb{N}$, let $\mathcal{P}_j$ be the set of all $\eta$ in $\mathcal{P}(\mathbb{N})$ such that the solution for the Kolmogorov equation (\ref{EqKolmogorov}) starting at $\eta$ is in $V$ for all times $ t \ge j$. By Lemma \ref{LemmaConvergenceMFE}, $\mathcal{P}(\mathbb{N})= \cup_j \mathcal{P}_j$ and since $\mathcal{P}_j \subset \mathcal{P}_{j+1}$, there is $k$ such that $  \Pi_0^\infty (\mathcal{P}_k) > 1- \epsilon $. It follows that
\begin{align}
\Pi_0^\infty (V) =  \Pi_k^\infty (V) =  \Pi^\infty (R_k \in V) \ge \Pi^\infty (R_0 \in \mathcal{P}_k)= \Pi_0^\infty (  \mathcal{P}_k ) > 1- \epsilon. \nonumber
\end{align}
Because $\epsilon$ and $V$ are arbitrarily chosen, $\Pi_0^\infty ( \{ \gamma_{\mu} \} )=1$.

{\bf Step 3}. Since $\bar{Q}_0^N$ converges in law to $\gamma_{\mu}$, by Theorem \ref{ThmChaoticityMeanFieldLimit}, the conclusions follow.

\bibliographystyle{acmtrans-ims}
\bibliography{BibSupermarketModel}
\end{document}